\documentclass[10pt,reqno,tbtags]{amsart}

\usepackage{caption}

\usepackage{tikz}
\usepackage{tikz-feynman}
\usetikzlibrary{calc,matrix,patterns}
\tikzset{  photon/.style={decorate, decoration={snake}, draw},
    electron/.style={draw, postaction={decorate},
        decoration={markings,mark=at position .55 with {\arrow[draw]{>}}}},
    gluon/.style={decorate, draw=magenta,
        decoration={coil,amplitude=4pt, segment length=5pt}} }

\usepackage[top=4cm , bottom=3cm, left=3cm , right=3cm]{geometry}
\usepackage{amsmath,amssymb,amsfonts,amsthm} 
\usepackage{mathrsfs,latexsym} 
\usepackage{mathtools}
\usepackage{thmtools}

\numberwithin{equation}{section}

\usepackage{calrsfs}
\DeclareMathAlphabet{\pazocal}{OMS}{zplm}{m}{n}

\usepackage{color}
\usepackage{cite}

\usepackage{epsfig}
\usepackage{graphicx}
\usepackage[all]{xy}

\usepackage{bm}
\usepackage{hyperref}
\DeclareFontFamily{OT1}{pzc}{}
\DeclareFontShape{OT1}{pzc}{m}{it}{<-> s * [1.1500] pzcmi7t}{}
\DeclareMathAlphabet{\mathpzc}{OT1}{pzc}{m}{it}

\usepackage{float}
\usepackage{wasysym}
\usepackage{slashed}
\usepackage{cite}\medskip
\usepackage{todonotes}
\roman{enumi}
\newtheorem{theorem}{Theorem}[section]
\newtheorem{definition}[theorem]{Definition}
\newtheorem{lemma}[theorem]{Lemma}
\newtheorem{proposition}[theorem]{Proposition}

\declaretheorem[sibling=theorem,style=definition, qed=\bell]{remark}

\numberwithin{equation}{section}

\newcommand{\CC}{{\mathbb C}}
\newcommand{\RR}{{\mathbb R}}
\newcommand{\MM}{{\mathbb M}}
\newcommand{\NN}{{\mathbb N}}

\newcommand{\ZZ}{{\mathbb Z}}

\newcommand{\Ac}{{\mathcal{A}}}
\newcommand{\Oc}{{\mathcal{O}}}
\newcommand{\Ocal}{{\mathcal{O}}}

\newcommand{\Dc}{{\mathcal{D}}}
\newcommand{\Ec}{{\mathcal{E}}}
\newcommand{\Ecal}{{\mathcal{E}}}
\newcommand{\Fc}{{\mathcal{F}}}
\newcommand{\Gc}{{\mathcal{G}}}

\newcommand{\Mcc}{{\mathcal{M}}}

\newcommand{\fA}{{\mathfrak A}}
\newcommand{\fB}{{\mathfrak B}}
\newcommand{\fD}{{\mathfrak D}}

\newcommand{\fG}{{\mathfrak G}}

\newcommand{\Floc}{\Fc_{\mathrm{loc}}}                   % domains of operators

\newcommand{\supp}{{\mathrm{supp} \, }}

\newcommand{\vect}[1]{\vec{#1}}
\newcommand{\eps}{\epsilon}
\newcommand{\ga}{\gamma}
\newcommand{\dl}{\delta}
\newcommand{\Dl}{\Delta}
\newcommand{\la}{\lambda}
\newcommand{\ox}{\otimes}
\newcommand{\x}{\times}
\newcommand{\ovl}{\overline}
\newcommand{\pp}{\slashed{\partial}}
\newcommand{\dg}{\mathrm{dg}}

\newcommand{\dd}{\slashed{D}}
\renewcommand{\S}{\slashed{S}}

\newcommand{\be}{\begin{equation}}
\newcommand{\ee}{\end{equation}}
\def\eg{{\it e.g.}}
\def\ie{{\it i.e.}}

\newcommand{\0}{\emptyset}

\newcommand{\vr}{\mathrm{v}}

\newcommand{\Ind}{\mathcal{P}_{\mathrm{finite}}(\NN)}

\newcommand{\etalchar}[1]{$^{#1}$}
\begin{document} 
%\title{On Noether's Theorem in Quantum Field Theory}
\title[C*-algebras and fermions]{C*-algebraic approach to interacting quantum field theory: Inclusion of Fermi fields}

\author[R. Brunetti]{Romeo Brunetti}
\address{Dipartimento di Matematica, Universit\`a di Trento, 38123 Povo (TN), Italy}
\email{romeo.brunetti@unitn.it}

\author[M. D\"utsch]{Michael D\"utsch}
\address{Institute f\"ur Theoretische Physik, Universit\"at G\"ottingen, 37077 G\"ottingen, Germany}
\email{michael.duetsch3@gmail.com}

\author[K. Fredenhagen]{Klaus Fredenhagen}
\address{II. Institute f\"ur Theoretische Physik, Universit\"at Hamburg, 22761 Hamburg, Germany}
\email{klaus.fredenhagen@desy.de}

\author[K. Rejzner]{Kasia Rejzner}
\address{Department of Mathematics, University of York, YO10 5DD York, UK}
\email{kasia.rejzner@york.ac.uk}

\date{\today}

\begin{abstract}
We extend the C*-algebraic approach to interacting quantum field theory, proposed recently by Detlev Buchholz and one of us (KF) to Fermi fields. The crucial feature of our approach is the use of auxiliary Grassmann variables in a functorial way.
\end{abstract}

\maketitle

\section{Introduction}
In a recent paper \cite{BF19} it was shown that the formal S-matrices (as generating functionals of time ordered products) generate a net of local C*-algebras which form a Haag-Kastler net. The S-matrices are there interpreted as local operations labeled by classical interaction Lagrangians, and it was shown that a few relations involving relativistic causality and a classical Lagrangian yield a structure which contains the canonical commutation relations and allows the construction of Haag-Kastler nets for quite general interactions. Let us briefly recall the technical steps: we consider a real classical scalar field $\phi$ and use its full configuration space, namely $\Ec\equiv C^\infty(M,\RR)$ over some globally hyperbolic spacetime $\Mcc=(M,g)$ (in the jargon of physicists, we are ``off-shell'' \ie\ we are not restricting to configurations which are \emph{solutions} of equation of motion). The space of observables $\Fc(M)$ is considered to be the linear space of 
local functionals over $\Ec$ of polynomial kind, \eg
$$
F[\phi]=\sum_{k=0}^N \int_M f_k(x) \phi(x)^k\ , 
$$ 
with compactly supported smooth densities $f_k$ on $M$. The support of the functionals ($\supp F$) is defined as the union of the supports of the test densities $f_k,k\ge1$, and hence they are all compactly supported on $\Mcc$. When $N>2$ these functionals describe 
local self-interactions of the field $\phi$. This allows us to introduce (interacting) Lagrangian densities
$$
M\ni x\mapsto L(x)[\phi] = \left(\frac{1}{2}\left(g(d\phi(x),d\phi(x))+m^2 \phi(x)^2\right) - \sum_{k=0}^N g_k \phi(x)^k\right)d\mu_g(x) \ ,
$$
where $m^2\ge 0$ and $g_k$ are real numbers (coupling constants), and then consider full Lagrangians $L$ as $\Fc$-valued maps, namely,
$$
\Dc(M)\ni f\mapsto L(f)[\phi]\doteq\int_M  L(x)[f\phi]\ .
$$

The, in general ill defined, global action functional (i.e., corresponding to $f\equiv 1$),  is replaced by the introduction of relative Lagrangians, \ie\ by defining
$$
\delta L(\phi_0)[\phi] \doteq L(f_0)[\phi]^{\phi_0} - L(f_0)[\phi] \doteq  L(f_0)[\phi+\phi_0] - L(f_0)[\phi]\ ,
$$
where $\phi_0\in\Ec_0\subset\Ec$ is compactly supported and $f_0\in\Dc(M)$ such that $f_0\equiv 1$ on $\supp\phi_0$. Note that the relative Lagrangians do not depend upon the choice of $f_0$ and belong to $\Fc(M)$, since in the subtraction the kinetic terms disappear and the linear terms with derivatives of the field $\phi$ can be written as linear terms in $\phi$ after integration by parts.

We have now all ingredients for the core construction:  we fix one such Lagrangian $L$ and define abstractly unitary symbols $S(F)$ labelled over $\Fc(M)$. We may interpret these symbols as formal S-matrices (scattering matrices: justifications for this interpretation can be found in \cite{BF19,BDFR21}). We then generate freely a group $\Gc_L$ out of these symbols modulo the following relations
\begin{itemize}
\item $S(F)S(\delta L(\phi_0))=S(F^{\phi_0}+\delta L(\phi_0))=S(\delta L(\phi_0))S(F)$, $\phi_0\in\Ec_0$, $F\in\Fc(M)$,
\item $S(F+G+H) = S(F+G)S(G)^{-1}S(G+H)$, $F,G,H\in\Fc(M)$, provided $\supp F$ is causally later than $\supp H$.
\end{itemize}
Notice that the first requirement corresponds to the incorporation of an equation of motion w.r.t.~to the total action given by $L$ plus $F$ (unitary version of the Schwinger-Dyson equation) and the second enforces a causality notion in $\Gc_L$ implied by the causality properties of spacetime. It is now a classical construction to pass from the group $\Gc_L$ to a group algebra $\Ac_L$ and moreover to show  \cite{BF19} that the last can be promoted to a C*-algebra (and we use the same symbol for both). It is a first gratifying surprise to discover that in \cite{BF19} one shows that the C*-Weyl algebra of the canonical commutation relations is contained as a proper C*-subalgebra in $\Ac_L$. Actually, one can do more by localization of the functionals, namely we can redo the construction for any open bounded (non empty) subregion $\Oc$ of $\Mcc$ and prove that the association $\Oc\mapsto \Ac_L(\Oc)$ is a Haag-Kastler net of C*-algebras (see again \cite{BF19} where this is shown for Minkowski spacetime\footnote{The generalization to any fixed globally hyperbolic spacetime $\Mcc$ is straightforward. However, the step towards local covariance requires some non trivial arguments which can be found in \cite{BDFR21}.}).

The formalism just described was restricted to scalar fields. It is the main goal of the present paper to generalize it to interacting Fermi fields.

Classical functionals for Fermi fields can be considered as linear functionals on the Grassmann algebra over the space of field configurations (Section \ref{sec:func}, see also \cite{Rej11a}). Following the construction recalled above, one would like to associate to each such functional a formal S-matrix, but only even functionals have a direct interpretation as arguments of formal S-matrices. On the other hand, the restriction to even functionals does not allow to formulate the unitary version of the Schwinger-Dyson equation, by which the classical Lagrangian enters the framework. 

There is a well known way out, namely the use of auxiliary Grassmann variables (the so-called $\eta$-trick, see, \eg, \cite{Due19,IZ}). 
% These auxiliary variables should not enter into the discussion of physical properties of the theory, they are only a useful device that helps to reduce the technical discussion to known properties, namely, they are needed to shift the combinatorics towards the bosonic situation. For our purpose, a finite number of Grassmann parameters suffices, but since there are plenty to choose, it turns out to be crucial that the action of the generated Grassmann algebra is functorial in the sense that all operations commute with homomorphisms between finite dimensional Grassmann algebras. (See \cite{Ll,HHS} for an extensive discussion.)  
These auxiliary variables are needed for shifting the combinatorics to the bosonic situation, but besides this they should not influence the structure of the theory. A finite number of Grassmann parameters is always sufficient, but nothing should depend on their choice.   Therefore the action of the generated Grassmann algebra should be functorial in the sense that all operations commute with homomorphisms between finite dimensional Grassmann algebras. (See \cite{Ll,HHS} for an extensive discussion.) Moreover, even linear relations between such homomorphisms should be respected, so that the embedding of the auxiliary variables into the theory does not change any relations between them.
We prove that such a covariant action of Grassmann variables on algebras can always be embedded into a tensor product of the Grassmann algebra with a uniquely determined algebra (Section \ref{sec:Grass}). 

We then present an adapted version of the axioms of \cite{BF19} in Section \ref{sec:Axioms} and show that they imply for the free Dirac field the canonical anticommutation relations (Section \ref{sec:CAR}).

This is used for solving another problem, namely the construction of a net of C*-algebras. Due to the fact that odd elements of a Grassmann algebra are nilpotent, it is not possible
to equip the tensor product of a nontrivial Grassmann algebra with the algebra $\fA$ of quantum fields with a C*-norm. Moreover, for the same reason, S-matrices of functionals which depend on these Grassmann variables have an expansion in polynomials in Grassmann variables with coefficients in $\fA$. There is no reason to expect that these coefficients have to be bounded, in general.
Instead one applies the abstract construction of the C*-algebra first to the subalgebra generated by S-matrices of even functionals, and adjoins then the smeared Dirac fields which are bounded due to the anticommutation relations (Section \ref{sec:C-star}).

Our construction avoids a famous no go theorem of Powers \cite{Pow67}, who proved that in dimension $>2$ canonical anticommutation relations for time zero fields are incompatible with interactions. Powers showed that the boundedness of canonical fermi fields together with causal (anti-)commutation relations imply the boundedness of time derivatives of fermi fields which then leads to vanishing of interaction under rather general conditions. The construction of interacting theories in terms of S-matrices as described above, however, does not involve the time zero fields and also does not provide, in the interacting case, information about a possible restriction of fields to a Cauchy surface. In the free case, such an information is obtained from the unitary Schwinger Dyson equation and yields the canonical anticommutation relations for the time zero fields, as shown in Section \ref{sec:CAR}.   

In Section \ref{sec:Pert} we check that our axioms are satisfied in renormalized perturbation theory. In the appendix we briefly describe the modifications which occur when both, Bose and Fermi fields, are present.
%\todo[inline]{Emphasize more the problem of the $C^*$-structure.}
%%%%%%%%%%%%%%%%%%
\section{Fermionic functionals}\label{sec:func}
A (local or nonlocal) fermionic functional on some real vector space $V$ is a linear form on the Grassmann algebra $\Lambda V$ over $V$.
Equivalently it is a sequence $F=(F_n)_{n\in\NN_0}$ of alternating $n$-linear forms on $V$ with
\be\label{eq:refF}
F(\vr_1\wedge\dots\wedge \vr_n)=F_n(\vr_1,\dots,\vr_n)\ , \quad F(1_{\Lambda V})=F_0\in\RR\ .
\ee
The pointwise product of fermionic functionals is defined by
\be\label{wedge}
\begin{split}
&(F\cdot G)_n(\vr_1,\dots,\vr_n)\\
&\quad=\sum_{\sigma\in S_n}\mathrm{sign}(\sigma)\sum_{k=0}^n \frac{1}{k!(n-k)!}F_k(\vr_{\sigma(1)},\dots,\vr_{\sigma(k)})G_{n-k}(\vr_{\sigma(k+1)},\dots,\vr_{\sigma(n)} )\ .
\end{split}
\ee
Let now $V$ be the space of functions on some topological space $T$, and let $F$ be a fermionic functional on $V$. 
The support of $F$ is defined by
\be\nonumber
\supp F=\{x\in T|\text{ for all neighborhoods }U\text{ of }x\ \exists n\in\NN,\vr_1,\dots,\vr_n\in V\ee
\be\label{support}\text{ with }\supp \vr_1\subset U\text{ such that }F_n(\vr_1,\dots,\vr_n)\neq0\}
\ee
A fermionic functional $F$ is called additive if it satisfies for all $n$ the condition
\be
\begin{split}\label{eq:additive}
F_n(\vr_1+w_1+&z_1,\dots,\vr_n+w_n+z_n)\\
&=F_n(\vr_1+w_1,\dots,\vr_n+w_n)-F_n(\vr_1,\dots,\vr_n)+F_n(\vr_1+z_1,\dots,\vr_n+z_n)
\end{split}
\ee
if $\supp(w_1,\dots,w_n)\cap\supp(z_1,\dots,z_n)=\0$. We remind that $\supp(w_1,\dots,w_n)=\cup_{j=1}^n\supp(w_j)$.

Consider now the special case where $V=\Gamma(M,E)$ is the space of sections of some vector bundle $E$ over the smooth manifold $M$, equipped with its natural Fr\'echet topology. Since we are now talking about topological vector spaces, we need to specify the topology for the tensor product $\Lambda^k V$. Fortunately, in the case we consider, $V$ is nuclear, so all the tensor products are equivalent. The appropriate notion of alternating $k$-linear continuous forms in this case is the topological dual of the completed tensor product $\widehat{\Lambda}^k V$, which turns out to be the completion
of $\Lambda^k \Gamma'(M,E)$ with respect to the topology of 
$\Gamma'(M,E)^{\hat{\otimes}k}\cong\Gamma'(M^k,E^{\boxtimes k})$
where all the duals are strong. This completion is the space of compactly supported antisymmetric distributional sections of the vector bundle $E^{\boxtimes k}$ over $M^k$. We denote it by $\Ocal^k(V[1])$ where the number in the square brackets denotes the degree shift (meaning that all the elements are understood to be in degree 1) and $\Ocal$ means the space of functions, so $\Ocal^k(V[1])$ is understood as a space of functions on the graded manifold $V[1]$.
%For $V=\Gamma(M,E)$, $\Gamma'(M,E)^{\hat{\otimes}k}$ 

We define the \textit{smooth fermionic functionals} as 
\be
\Ocal(V[1])\doteq  \prod_{k=0}^\infty \Ocal^k (V[1])\,,
\ee
where $\Ocal^0 (V[1])\equiv \CC$.
%Our notation emphasizes the fact that this space can be understood as the space of functions on the graded manifold $V[1]$. 
An element $F\in \Ocal(V[1])$ will be represented by the sequence $(F_n)_n$, where $F_n\in \Ocal^n(V[1])$. Note that, due to the required continuity,  smooth fermionic functionals are always compactly supported, in contrast to the bosonic case (cf. \cite{BFR}). They are also always differentiable in the following sense:
%Let us now discuss the notion of derivative for the type of functionals 
%introduced above. 
\begin{definition}\label{GradedDerivaitve}
	Let $F\in\Ocal^k(V[1])$, $h\in V^{\hat\otimes k-1}$, $\vect{h}\in V$. 
	The left derivative of $F$ at $h$ in the direction of $\vect{h}$ is defined,
	for every integer $k\geq 0$, 
	\begin{eqnarray}
	%	\left<F^{(1)}(h),\vect{h}\right>
		\left<\vect{h},F^{(1)}(h)\right>
		&=& F(\vect{h}\wedge h),\quad \text{for}\quad  k>0,\\
		F^{(1)} &=& 0\quad F\in\Ocal^0(V[1])\,.
	\end{eqnarray}
	This definition is then extended to $\Ocal(V[1])$ in a natural way. The right derivative is defined analogously.
\end{definition}
To illustrate this definition, consider the case  
$M=\mathbb{M}$ of Minkowski spacetime, 
$E=\mathbb{M}\times \RR$ and $V=\Gamma(\mathbb{M},E)$. We define
$F\in\Ocal^2(V[1])$ by
\be
F(h_1\wedge h_2)=\sum_{\mu,\nu=0}^3\int f(x) 
a^{\mu\nu}\partial_\mu h_1(x)\partial_\nu h_2(x)\,
d^4x\,,
\ee
where $h_1$ and $h_2$ are in $\Gamma(\mathbb{M},E)=C^\infty(\mathbb{M})$,
$f$ is in ${C}^{\infty}_0(\mathbb{M},\CC)$ and $a$ is any antisymmetric, constant $4\times 4$ matrix. Now we have,
for $h$ and $\vect{h}$ in $\Gamma(\mathbb{M},E)$:
\be
\left<\vect{h},F^{(1)}(h)\right>
%\left<F^{(1)}(h),\vect{h}\right>
= \sum_{\mu,\nu=0}^3\int f(x)
a^{\mu\nu}(\partial_\mu \vect{h}\partial_\nu h)(x)\,d^4x\,.
\ee
%[A concrete example is:\todo{Can this concrete example be omitted?}
%\]
%so
%\[
%\left< F^{(1)}(\vr),h\right>= \int f(x) 
%\big(\partial_1 h(x)\partial_2 \vr(x) - \partial_2 h(x)\partial_1 \vr(x)\big) d^4x\,.\quad]
%\] 
As a second example,
take $M=\mathbb{M}$, $E=\mathbb{M}\times \RR^k$ and again $V=\Gamma(\MM,E)$.
Let $h_1=(h_1^j)_{j=1}^k$, $h_2=(h_2^j)_{j=1}^k$ and $h= (h^j)_{j=1}^k$ be three sections
in $\Gamma(\mathbb{M},E)$. 
Define 
\be
G(h_1\wedge h_2)=\sum_{i,j=1}^k\int a_{ij}(x)\,h_1^i(x)h_2^j(x)\,\,d^4x\,,
\ee
with any antisymmetric $k\times k$ matrix $(a_{ij}(x))$, all coefficients satisfying $a_{ij}\in {C}^{\infty}_0(\MM,\CC)$. We obtain
\be
%\left<G^{(1)}(h),\vect{h}\right>
\left<\vect{h},G^{(1)}(h)\right>= \sum_{i<j}\int a_{ij}(x)\,\left(\vect{h}^i(x)h^j(x)-\vect{h}^j(x)h^i(x)\right)\, d^4x.
\ee
It has been shown, see e.g. \cite{Rej11a} that the left derivative defined this way satisfies the Leibniz rule. Iterating this definition, we can define the $n$th left derivative $F^{(n)}$ of a fermionic functional.

Note that the derivative of  $F\in\Ocal^k(V[1])$ is a jointly continuous map 
\be
F^{(1)}: %V^{\hat\otimes k-1}\times V\rightarrow \RR\ .
V\times V^{\hat\otimes k-1}\longrightarrow \RR\ .
\ee
 It can be identified with a vector-valued distribution in $\Gamma'(M,E)\hat\otimes \Ocal^{k-1}(V)$. More generally, the $n$th derivative  $F^{(n)}$ is an element of $\Gamma'(M^n,E^{\oplus n})\hat\otimes \Ocal(V[1])$. The completed tensor product used here is the projective tensor product.
For more details, see e.g. section 3.3 of \cite{Book}. As noted in \cite{Rej11a}, the definitions of a wavefront set can be extended to such vector-valued distributions and the usual theorems about multiplying distributions apply to this case.

The ``standard'' characterization of locality for a compactly supported functional 
$F\in\Ocal^k(V[1])$
is the requirement that $F$  has the form 
%$(\vr_1,\dots,\vr_k)\in V^k$ there exists 
%$i_1,\dots,i_k\in\mathbb{N}$ such that
\begin{equation}\label{eq:F-local}
F(h_1,\dots,h_k)=\int_{M} \alpha(
j_x(h_1),\dots,j_x(h_k))\,,
\end{equation}
where $\alpha$ is a compactly supported density-valued alternating function on $k$ arguments from the jet bundle. Note that $\alpha$ automatically depends only on the finite jet of the arguments, due to multilinearity and continuity.
%\todo{proposal for reformulation (KF)}

It is easy to see that 
every local functional \eqref{eq:F-local} is additive \eqref{eq:additive}; however, additivity does not suffice for locality -- 
an additional smoothness assumption is needed. For the analogous problem for bosonic functionals, locality is proved for two different 
versions of this additional assumption, see \cite[Thm.~VI.3]{BDGR} and \cite[Prop.~2.2]{BFR}). We give here the fermionic analogon
of the former theorem, the general case of functionals depending on both fermionic and bosonic variables
is treated in the appendix.
%\todo{cite [BFR]}
\begin{theorem}
	Let 
	$F\in\Ocal(V[1])$.	Assume that
	\begin{enumerate}
		\item $F$ is additive.
		\item For every  $h\in\bigoplus_{k\in\NN} V^{\hat{\otimes}k}$, the first derivative
		$F^{(1)}$ of $F$ has
		empty  wave front set as a vector-valued distribution and the map
		$h\mapsto F^{(1)}(h)$ is Bastiani smooth\footnote{See \cite{Mich38,Bas64,Ham} for details on this notion of differentiability and smoothness of functionals on locally convex topological vector spaces, and \cite{BDGR} for a pedagogical review.} from $\bigoplus_{k\in\NN} V^{\hat{\otimes} k}$  to $\Gamma_c(M,E^*)$. Here $E^*$ denotes dual bundle.
	\end{enumerate}
	Then $F$ is local.
	%Then, for every  $u\in\bigoplus_{k\in\NN} E^{\hat{\otimes} k}$
	%[or rather: ... for every $\vr_1\otimes\dots\otimes \vr_k\in %V^{\hat{\otimes}k}$\todo{???}], there is an integer $N$ and a smooth 
	%$\CC$-valued function $\alpha$ on the $N$-jet bundle such that 
	%\begin{equation}
	%F(\vr_1\otimes\dots\otimes \vr_k)=\int_{M} %\alpha(j^{i_1}_x(\vr_1),\dots,j^{i_k}_x(\vr_k))\,,
	%\end{equation}
	%and $i_1,\dots,i_k<N$.
\end{theorem}
\begin{proof}
The proof is patterned after the paper \cite{BDGR}, and we provide the necessary ideas to fill the gaps for the use of their results in our context.
	
	Let $F\in \Ocal^k(V[1])$, $k\neq 0$. We have
	\be
	\begin{split}
	F(h_1\wedge\dots\wedge h_k) &=  \frac{1}{k}\sum_{i=1}^k(-1)^{k-1}\int_M   
	F^{(1)}(h_1\wedge\dots\widehat{h_i}\dots 
	\wedge h_k)(x) h_i(x)\,\, dx\\
	&= \int_M   F^{(1)}(h_2\wedge\dots \wedge h_k)(x) h_1(x)\,\, dx\,,
	\end{split}
	\ee
	Denote $h\doteq  h_1\wedge\dots\wedge h_k$ and write
	\be
	F(h)=\int_M c_{h}(x)\,\, dx\,,
	\ee
	where 
	$c_{h}(x)= \mathrm{ev}_x\left(F^{(1)}
	(h_2\wedge\dots \wedge h_k) h_1\right)$.	
	%	Note that $F\in L(V^k,\RR)$ is a bounded linear mapping and we know (see e.g. Proposition 5.2 from \cite{Michor}) that $L(\Ecal^k,\RR)\cong L(V^{k-1},L(V,\RR))$, where the topology is that of uniform convergence on bounded sets.\todo{Note: this is a bornological isomorphism.}
	Now, we use the fact that, by assumption, the wavefront set of $F^{(1)}$ is empty and the map $h\mapsto F^{(1)}(h)$ is Bastiani smooth, to apply proposition VI.14 of \cite{BDGR} 
	and conclude that the function 
	$c_{h}$ depends only on
	finite jets of $h_1,\dots,h_k$. Finally, we use 
	Lemma VI.15 of the same reference and their Proposition VI.4 to conclude that the resulting 
	function  $\alpha$ on the jet bundle is smooth. Hence, $F$ is local.
\end{proof}
%We choose now $V=\Ec$,  the space of field configurations and its subspace $\Ec^c$ of compactly supported field configurations. 
%%%%%%%%%%%%%%%%%
\section{Covariant Grassmann multiplication}\label{sec:Grass}
We are confronted with the following problem: We want to construct the algebra of observables, extended also to fermionic operators. But the relations characterizing this algebra $\fA$ contain auxiliary Grassmann parameters whose only purpose is to allow the use of combinatorial formulas known from the bosonic case. We thus obtain in a first step subalgebras  $\fA_G$ of tensor products  $G\otimes \fA$ of Grassmann algebras $G$ with $\fA$ that are generated by even elements and the Grassmann algebra itself (understood as $G\otimes 1_{\fA}$). 
The aim is to reconstruct the algebra $\fA$ from that family of subalgebras. To this end we equip this family of subalgebras with the following structure.

Let $\mathfrak{Grass}$ denote the category of finite dimensional real Grassmann algebras, with homomorphisms as arrows and let $\mathfrak{Alg}^{\ZZ_2}$ be the category of $\ZZ_2$-graded unital associative algebras, with unital homomorphisms respecting the $\ZZ_2$ gradation as arrows. Let now $R:\mathfrak{Grass}\rightarrow \mathfrak{Alg}^{\ZZ_2}$ be the inclusion functor.

\begin{definition}\label{def:Grass-functor}
A covariant Grassmann multiplication algebra is pair $(\mathfrak{G},\iota)$ consisting of a
functor \be
\fG:\mathfrak{Grass}\to\mathfrak{Alg}^{\ZZ_2}
\ee
and a natural embedding $\iota:R\Rightarrow \mathfrak{G}$ \ie\ a family $(\iota_G)_G$ of injective homomorphisms $\iota_G:G\to \fG G$ with 
\be
\iota_{G'}\circ\chi=\fG\chi\circ\iota_G\ ,\quad \text{ for homomorphisms }\chi:G\to G'\ .
\ee
%\begin{figure}[htb]
%\centering
$$
\begin{tikzpicture}
%\matrix(M)[matrix of math nodes, row sep=2em, column sep=2em]{
%& G &&&&&& G' &\\
%&&&&&&&&\\
%G &&&&&& G' &&\\
%&& \fG G &&&&&& \fG G'\\
%   };
%\begin{scope}[every node/.style={midway,font=\footnotesize}]
%\draw[->] (M-1-2) -- node[above] {$\chi$} (M-1-8) ;
%\draw[->] (M-1-2) -- node[above left] {$\mathrm{id}_{\mathfrak{Grass}}$} (M-3-1) ;
%\draw[->] (M-3-1) -- node[above right] {$\chi$} (M-3-7);
%\draw[->] (M-3-1) -- node[below left] {$\iota_G$} (M-4-3) ;
%\draw[->] (M-1-2) -- node[above right] {$\fG$} (M-4-3) ;
%\draw[->] (M-4-3) -- node[below] {$\fG\chi$} (M-4-9);
%\draw[->] (M-1-8) -- node[above left] {$\mathrm{id}_{\mathfrak{Grass}}$} (M-3-7);
%\draw[->] (M-3-7) -- node[below left] {$\iota_{G'}$} (M-4-9) ;
%\draw[->] (M-1-8) -- node[above right] {$\fG$} (M-4-9) ;
%\end{scope}
\matrix(M)[matrix of math nodes, row sep=1.5em, column sep=1.5em]{
G &&& G' \\
&&&\\
\fG G &&& \fG G'\\
   };
\begin{scope}[every node/.style={midway,font=\footnotesize}]
\draw[->] (M-1-1) -- node[above] {$\chi$} (M-1-4);
\draw[->] (M-1-1) -- node[left] {$\iota_G$} (M-3-1) ;
\draw[->] (M-3-1) -- node[below] {$\fG\chi$} (M-3-4);
\draw[->] (M-1-4) -- node[right] {$\iota_{G'}$} (M-3-4) ;
\end{scope}
\end{tikzpicture}
$$
We require the following properties of $(\mathfrak{G},\iota)$:

\begin{enumerate}
%\caption{xy}
%\label{fg:1}
%\end{figure}
\item $\iota_G(G)$ is graded central in $\fG G$, in the sense that
\be\label{eq:eta-a}
\iota_G(\eta)\,a=(-1)^{\mathrm{dg}(\eta)\mathrm{dg}(a)}\,a\,\iota_G(\eta)\ ,\ \eta\in G,\ a\in\fG G\ ,
\ee
where $\dg(\cdot)\in\{0,1\}$ denotes the degree.\footnote{In the literature, often the degree in the Grassmann algebra and the degree of intrinsic fermionic variables are distinguished, such that intrinsic variables and auxiliary Grassmann parameters always commute. While this sometimes avoids sign factors in practical calculations (see e.g.~\cite[Chap.~5]{Due19})
it seems to be less appropriate in a conceptual analysis.}
\item \label{linear}Let $\lambda_i\in\RR$ and $\chi_i:G\to G'$, $i=1,\dots,n$ be homomorphisms between Grassmann algebras with $\sum_{i=1}^n\lambda_i\chi_i=0$. Then $\sum_{i=1}^n\lambda_i\,\fG\chi_i=0$.\footnote{This entails that $\fG$ is a functor between \emph{enriched} categories (over the category of vector spaces).}
\end{enumerate}
\end{definition}
The first property in the above definition is quite natural to require. 
The second one is a condition motivated by the specific problem we are trying to solve, namely, without this condition we would not be able to prove the key reconstruction result (\ie\ the reconstruction of the algebra $\fA$).
Indeed, the condition does not follow from the other conditions, as may be seen from the example $\fG G=G\otimes G$ and $\fG\chi=\chi\otimes \chi$. We observe that the linearity condition (\ref{linear}) may be understood as a minimality condition 
on the extension by anticommuting parameters.

An example of a covariant Grassmann multiplication algebra is the functor $\fG^{\fA}$ with a graded unital algebra $\fA$ 
which maps Grassmann algebras $G$ to tensor products 
$\fG^{\fA}G=G\otimes\fA$ with the product
\be
(\eta_1\ox a_1)\cdot (\eta_2\ox a_2)\doteq (-1)^{\dg(\eta_2)\dg(a_1)}\,(\eta_1\eta_2)\ox(a_1a_2),\quad
\eta_1,\eta_2\in G,\,\,a_1,a_2\in\fA\ ,
\ee
and morphisms $\chi:G\to G'$ to morphisms  
$\fG^{\fA}\chi:G\otimes\fA\to G'\otimes\fA$ by
\be\label{eq:GA}
\fG^{\fA}\chi(\eta\otimes a)=\chi(\eta)\otimes a\ ,\ \eta\in G\ ,\ a\in\fA\ .
\ee
The natural transformation $\iota$ is given by
\be
\iota_G(\eta)=\eta\otimes 1_{\fA}\ ,\ \eta\in G\ .
\ee
It is easy to see that also the linearity condition (\ref{linear}) of Definition \ref{def:Grass-functor} is satisfied.
In the following we simplify the notation by identifying $\iota_G(\eta)$ with $\eta$ for $\eta\in G$ and
$1_{G}\otimes a$ with $a$ for $a\in\fA$, and similarly we write $\eta a$ for $\eta\ox a\in G\ox\fA$.

 We apply this construction to the exterior algebra over some vector space $V$ (i.e. $\mathfrak{A}=\Lambda V$) as well as to its dual, the algebra of fermionic functionals on $V$. The latter we mainly restrict to the subspace of local functionals (denoted by $\Floc$),
such that $\fG^{\Floc}$ associates to every Grassmann algebra $G$ a $G$-bimodule. A fermionic functional induces, for any $G$, a  $G$-module homomorphism $F_G$ from $G\otimes \Lambda V$ to $G$ by
\be\label{eq:G-extension}
F_G(\omega\eta)=F(\omega)\eta=\eta F(\omega)\ ,\qquad \omega\in\Lambda V, \eta\in G\ \ ,
\ee
and we identify $\eta F$ with the map $\omega\mapsto \eta F(\omega)$.
The $\wedge$-symbol for the product in $\Lambda V$ is usually omitted. At some places we use it in order to make clear that $V$ is identified with $\Lambda^1(V)$.

As an example, for $\vr^1,\vr^2\in\Lambda^1(V)=V$ and  odd elements $\eta_1,\eta_2\in G$, we obtain
\be
F_G\bigl((\vr^1\eta_1) (\vr^2\eta_2)\bigr)=F_G\bigl((\vr^1\vr^2)(\eta_2\eta_1)\bigr)
=F(\vr^1\wedge\vr^2)\eta_2\eta_1\ .
\ee
The family $(F_G)_G$ is a natural transformation $\mathfrak{F}:\mathfrak{G}^{\Lambda V}\Longrightarrow\mathfrak{G}^{\RR}$, that is,
\be
\mathfrak{G}^{\RR}\chi\circ F_G=F_{G'}\circ\mathfrak{G}^{\Lambda V}\chi\ .
\ee
%\begin{figure}[htb]
%\centering
$$
\begin{tikzpicture}
%\matrix(M)[matrix of math nodes, row sep=2em, column sep=2em]{
%& G &&&&&& G' &\\
%&&&&&&&&\\
%G\ox \Lambda V &&&&&& G'\ox \Lambda V &&\\
%%&& G\ox\RR &&&&&& G'\ox\RR\\
%   };
%\begin{scope}[every node/.style={midway,font=\footnotesize}]
%\draw[->] (M-1-2) -- node[above] {$\chi$} (M-1-8) ;
%\draw[->] (M-1-2) -- node[above left] {$\fG^{\Lambda V}$} (M-3-1) ;
%\draw[->] (M-3-1) -- node[above right] {$\fG^{\Lambda V}\chi$} (M-3-7);
%\draw[->] (M-3-1) -- node[below left] {$F_G$} (M-4-3) ;
%\draw[->] (M-1-2) -- node[above right] {$\fG^\RR$} (M-4-3) ;
%\draw[->] (M-4-3) -- node[below] {$\fG^\RR\chi$} (M-4-9);
%\draw[->] (M-1-8) -- node[above left] {$\fG^{\Lambda V}$} (M-3-7);
%\draw[->] (M-3-7) -- node[below left] {$F_{G'}$} (M-4-9) ;
%\draw[->] (M-1-8) -- node[above right] {$\fG^\RR$} (M-4-9) ;
%\end{scope}
\matrix(M)[matrix of math nodes, row sep=1.5em, column sep=1.5em]{
G\ox \Lambda V &&& G'\ox \Lambda V \\
&&&\\
G\ox\RR &&& G'\ox\RR\\
   };
\begin{scope}[every node/.style={midway,font=\footnotesize}]
\draw[->] (M-1-1) -- node[above] {$\fG^{\Lambda V}\chi$} (M-1-4);
\draw[->] (M-1-1) -- node[left] {$F_G$} (M-3-1) ;
\draw[->] (M-3-1) -- node[below] {$\fG^\RR\chi$} (M-3-4);
\draw[->] (M-1-4) -- node[right] {$F_{G'}$} (M-3-4) ;
\end{scope}
\end{tikzpicture}
$$
%\caption{xy}
%\label{fg:1}
%\end{figure}

$F$ is already fixed if we know the maps $F_{G}$  on all elements of the form
\be
\exp {\sum_{i\in I} \vr^i\eta_i}
\ee
with odd elements $\eta_i\in G$, $\vr^i\in\Lambda^1(V)=V$ and a finite index set $I\in\Ind$, where $F_G(1_G)=F_0 1_G$ (see \eqref{eq:refF}).
(This is called the ``even rules principle'' in \cite{FieldsStrings,CCF,Ll}.) So fermionic functionals $F$ on $V$ can be characterized as coherent families of $G$-valued maps $F_G\circ\exp$ on the even part of the Grassmann modules $G\otimes V$.  

In particular we can define shifts in the arguments as they occur in the unitary Dyson-Schwinger equation (\ie, the relation 'Dynamics' given in \eqref{eq:dynamics} below). A shifted functional ${F}^{\vec w}$, with $\vec w=\sum_{j\in J} \vec w^j\theta_j$ with  
odd elements $\theta_j$ of some Grassmann algebra $G'$ and $\vec w^j\in V$, $J\in\Ind$, is defined as a family 
$(F^{\vec w}_{G})_{G}$ of $G$-module maps from $G\otimes\Lambda V$ to $G\otimes G'$, 
\be\label{eq:shift-F}
\begin{split}
F_{G}^{\vec w}\bigl(\exp{\sum_{i\in I} \vr^i\eta_i}\bigr) &=
F_{{G\otimes G'}}\bigl(\exp{(\sum_{i\in I} \vr^i\eta_i+\sum_{j\in J}\vec w^j \theta_j)}\bigr)\\
&=\sum_{n\geq 0}\sum_{i_1<\ldots <i_n}\,F_n^{\vec w}(\vr^{i_1},\ldots,\vr^{i_n})\eta_{i_n}\cdots\eta_{i_1}\ ,
\end{split}
\ee
with  alternating multilinear $G'$-valued maps $F^{\vec w}_n$ as components
%\be
%F(\vr_1+w,\dots,\vr_n+w)=\sum_{I\subset\{1,\dots,n\}}\sum_{|I|+|K|=n}\stackrel{\leftarrow}{\prod_{k\in K}}\eta_k\,F(w_k,k\in K,\vr_i,i\in I)
%\ee
%hence the $n$-th term of $F^{w}$ is
\be
F^{\vec w}_n(\vr^1,\dots,\vr^n)=\sum_{k\geq 0}\sum_{j_1<\ldots <j_k\in J}F_{k+n}(\vr^1,\dots,\vr^n,\vec w^{j_1},\dots,\vec w^{j_k})\,
\theta_{j_k}\cdots\theta_{j_1}\ .
\ee 

We will see that every covariant Grassmann multiplication algebra is almost of the form $\fG^{\fA}$ for some graded algebra $\fA$,
which is \emph{universal} in the following sense.
%uniquely determined up to an isomorphism. 
%More precisely, $\fA$ is obtained by the following universal construction:   
\begin{theorem}\label{thm:Grass}
Let $\fG$ be a covariant Grassmann multiplication algebra as defined above.
Then there exists a graded unital algebra $\fA$ and a natural embedding
\be
\sigma\equiv(\sigma_G)_G:\mathfrak{G}\Longrightarrow\mathfrak{G}^{\fA}
\ee
%(that is, $\sigma$ satisfies \eqref{eq:natur}) 
such that for any other graded unital algebra $\fA'$ with a natural embedding $\sigma':\mathfrak{G}\Longrightarrow\mathfrak{G}^{\fA'}$ there exists a unique homomorphism $\tau:\fA\rightarrow\fA'$ with $\sigma_G'=(\mathrm{id}\otimes\tau)\circ\sigma_G$.
\end{theorem}
%\begin{figure}[htb]
%\centering
$$
\begin{tikzpicture}
%\matrix(M)[matrix of math nodes, row sep=2em, column sep=5em]{
%& G &\\
%&&\\
%& \fA_G\doteq\fG G &\\
%G\ox \fA && G\ox \fA'\\
%   };
%\begin{scope}[every node/.style={midway,font=\footnotesize}]
%\draw[->] (M-1-2) -- node[above left] {$\fG^\fA$} (M-4-1) ;
%\draw[->] (M-1-2) -- node[above right] {$\fG^{\fA'}$} (M-4-3) ;
%\draw[->] (M-1-2) -- node[right] {$\fG$} (M-3-2) ;
%\draw[->] (M-3-2) -- node[below right] {$\sigma_G$} (M-4-1);
%\draw[->] (M-3-2) -- node[below left] {$\sigma'_G$} (M-4-3);
%\draw[->] (M-4-1) -- node[below] {$\mathrm{id}\ox\tau$} (M-4-3);
%\end{scope}
\matrix(M)[matrix of math nodes, row sep=2em, column sep=5em]{
& \fA_G\doteq\fG G &\\
G\ox \fA && G\ox \fA'\\
   };
\begin{scope}[every node/.style={midway,font=\footnotesize}]
\draw[->] (M-1-2) -- node[above left] {$\sigma_G$} (M-2-1);
\draw[->] (M-1-2) -- node[above right] {$\sigma'_G$} (M-2-3);
\draw[->] (M-2-1) -- node[below] {$\mathrm{id}\ox\tau$} (M-2-3);
\end{scope}
\end{tikzpicture}
$$
%\caption{xy}
%\label{fg:1}
%\end{figure}

\noindent The proof of the theorem will be splitted into four parts.

\begin{proof}[First part of the proof] We construct $\fA$ together with natural embeddings \be\sigma_G:\fA_G\doteq\fG G\to G\otimes\fA \ ,\ee \ie \ injective homomorphisms satisfying
\be\label{eq:natur}
\sigma_{G'}\circ\mathfrak{G}\chi=\mathfrak{G}^{\mathfrak{\fA}}\chi\circ\sigma_G
\ee
for homomorphisms $\chi:G\to G'$.

We use the fact that any finite dimensional real Grassmann algebra is isomorphic to $\Lambda\RR^n$ for some $n\in\NN_0$. In a first step we study the linear hull of homomorphisms from $\Lambda\RR^n$ to $\Lambda\RR^m$. Let $\eta_i,i=1,\dots,n$ denote the generators of $\Lambda\RR^n$ and $\theta_j,j=1,\dots,m$ the generators of $\Lambda\RR^m$.
Then $\{\eta_I,I\subset\{1,\dots,n\}\}$ with $\eta_I=\prod_{i\in I}\eta_i$ is a basis of $\Lambda\RR^n$, and  $\{\theta_J,J\subset\{1,\dots,m\}\}$ with $\theta_J=\prod_{j\in J}\theta_j$ is a basis of $\Lambda\RR^m$.
\begin{lemma}
Let $\chi$ be a linear map from $\Lambda\RR^n$ to $\Lambda\RR^m$ with 
\be\label{eq:chi-linear}
\chi(\eta_I)=\sum_{J\subset\{1,\dots,m\}}c_{IJ}\theta_J\ .
\ee
$\chi$ is a linear combination of homomorphisms of Grassmann algebras if and only if
\be
c_{IJ}=0
\ee
whenever $|I|+|J|$ is odd or $|J|<|I|$.
\end{lemma}
\begin{proof}
By definition, homomorphisms $\chi$ of Grassmann algebras preserve the degree mod 2, and $\chi(\eta_I)=\prod_{i\in I}\chi(\eta_i)$ has form degree at least $|I|$ if $\chi(\eta_I)\not=0$. This proves the \emph{only if} statement of the lemma.  

To prove the other direction we construct matrix units $E_{JI}$, $E_{JI}(\eta_{I'})=\delta_{II'}\theta_J$ for $|I|+|J|$ even and $|J|\ge|I|$ as linear combinations of homomorphisms. Obviously the given linear map $\chi$ \eqref{eq:chi-linear} can be written as
\be
\chi=\sum_{I\subset\{1,\dots,n\},\,J\subset\{1,\dots,m\}}c_{IJ}E_{JI}.
\ee
To show that $E_{JI}$ is a linear combination of homomorphisms of Grassmann algebras,
let $P_I$ be the homomorphism of $\Lambda\RR^n$ with $P_I\eta_i=\eta_i$ if $i\in I$ and $P_I\eta_i=0$ otherwise. Then
\be\label{eq:EI}
E_I\doteq P_I\prod_{i\in I}(\mathrm{id}-P_{I\setminus\{i\}})
\ee
projects onto the subspace of multiples of $\eta_I$. Given $I=\{i_1,\dots,i_{|I|}\}$ 
(with $i_1<\dots<i_{|I|}$) and $J$ with $|I|+|J|$ even and $|J|\geq |I|$,
let $(J_1,\dots,J_{|I|})$ be a partition of $J$ into odd subsets such that the indices in $J_k$ are smaller than those in $J_l$ if $1\le k<l\le|I|$, and consider the homomorphism $\chi^{JI}:\Lambda\RR^n\to\Lambda\RR^m$ with $\chi^{JI}(\eta_{i_k})=\theta_{J_k}$, 
$1\leq k\leq |I|$, and $\chi^{JI}(\eta_l)=0$ for $l\not\in I$. Hence, $\chi^{JI}(\eta_I)=\theta_J$. Then 
\be 
E_{JI}=\chi^{JI}\circ E_I
\ee
is a linear combination of homomorphisms.
\end{proof}
In the following we denote these matrix units by $E^{mn}_{JI}$ in order to indicate that they are mappings from $\Lambda\RR^n$ to $\Lambda\RR^m$; note that $E_I^n\doteq E_I$ \eqref{eq:EI} can be written as $E^n_I= E^{nn}_{II}$. 
Also the projections $P_I$ get an upper index $n$. Moreover, we extend the action of the functor $\fG$ to linear combinations of homomorphisms: $\fG(\sum_i\la_i\chi_i)\doteq\sum_i\la_i\,\fG\chi_i$.
We use the following notations:
\begin{equation}
    \begin{split}
        \pi_K&=\fG P^n_K\\
        \rho_K&=\fG E^n_K\\
        %\chi^n_J&=\fG E^{n,|J|}_{J,\{1,\dots,\|J|\}}\\
        %\chi_n^J&=\fG E^{|J|,n}_{\{1,\dots,\|J|\},J}
    \end{split}
\end{equation}
%We then construct a direct sum decomposition corresponding to the degree in Grassmann variables. We use the fact that any finite dimensional real Grassmann algebra is isomorphic to $\Lambda\RR^n$ for some $n\in\NN_0$. %We
%consider the homomorphisms $\pi_K:\fA_{\Lambda\RR^n}\to\fA_{\Lambda\RR^n}$, for $K\subset\{1,\dots,n\}$, induced by the action %$\eta_i\mapsto \eta_i$ for $i\in K$ and $\eta_i\mapsto 0$ for $i\not\in K$ on the generators of $\Lambda\RR^n$. 
%For $K\subset\{1,\dots,n\}$ let $\pi_K:\fA_{\Lambda\RR^n}\to\fA_{\Lambda\RR^n}$ be as in Definition \ref{def:Grass-functor}.
%\pi_K$ is a projection, and its image is isomorphic to $\fA_{\Lambda\RR^{|K|}}$. 
The projections $\pi_\bullet$ satisfy the relation
\be\label{eq:projections}
\pi_K\pi_J=\pi_{K\cap J},
\ee
which shows that they commute with each other, and the definition \eqref{eq:EI} turns into
%We now form the projections%
%\footnote{This definition of $\rho_K$ generalizes the definition of $\rho_{\{1,\ldots,n\}}$ given in \eqref{eq:rho-n}.}
\be\label{eq:projectionsrho}
\rho_K=\pi_K\prod_{k\in K}(\mathrm{id}-\pi_{K\setminus\{k\}})\ .
\ee
The projections $\rho_\bullet$ form a direct sum decomposition of $\fA_{\Lambda\RR^n}$:
\begin{lemma} The projections $\rho_\bullet$ 
%defined in \eqref{eq:projectionsrho} 
have the following properties:
\begin{itemize}
\item[$(i)$] Direct sum decomposition
\be
\rho_K\,\rho_J =\delta_{JK}\,\rho_K \label{eq:orthogonal}\ ,
\ee
\be
\sum_{K\subset\{1,\ldots,n\}}\rho_K =\mathrm{id}_{\fA_{\Lambda\RR^n}}\ . \label{eq:completeness}
\ee
\item[$(ii)$] Convolution
\be
\rho_K(ab) =\sum_{J\subset K}\rho_J(a)\rho_{K\setminus J}(b)\label{eq:product} \ .
\ee
\end{itemize}
\end{lemma}
\begin{proof} 
%For $K=J$, equation \eqref{eq:orthogonal} follows from the fact that $\rho_K$ is a projection as a product of commuting projections.
%If $K\neq J$ we choose some $i\in K\cup J$ which is not  in $K\cap J$. We may assume that $i\in K$. But then
%\be
%\rho_K\rho_J=\pi_{J\cap K}(\mathrm{id}-\pi_{K\setminus\{i\}}) \prod_{k\in K\setminus\{i\}}
%(\mathrm{id}-\pi_{K\setminus\{k\}})\prod_{j\in J}(\mathrm{id}-\pi_{J\setminus\{j\}})=0
%\ee
%since 
%\be\pi_{J\cap K}\pi_{K\setminus\{i\}}=\pi_{J\cap K}\ .\ee
%To prove \eqref{eq:completeness} we expand the product in the definition of $\rho_K$ and find
%\be\label{eq:expansion}
%\begin{split}
%\sum_K\rho_K &=\sum_K\pi_K\sum_{J\subset K}(-1)^{|K\setminus J|}\pi_J=\sum_J\pi_J\sum_{K\supset J}(-1)^{|K\setminus J|}\\
%&=\sum_J\pi_J(1-1)^{n-|J|}=\pi_{\{1,\dots,n\}}=\mathrm{id}\ ,
%\end{split}
%\ee
%where we used  \eqref{eq:projections}. % in the second step. \eqref{eq:projections} is also used for the expansion!
(i) Since $(E_K^n)_{K\subset\{1,\ldots,n\}}$ is precisely the set of projections onto the one dimensional subspaces of $\Lambda\RR^n$ corresponding to the basis $(\eta_K)_{K\subset\{1,\ldots,n\}}$, they satisfy $E_K^n\, E_J^n =\delta_{JK}\,E_K$ and
$\sum_{K}E_K^n =\mathrm{id}_{\Lambda\RR^n}$. Under application of the functor $\fG$, these relations are maintained; 
in particular, by definition of a functor it holds that $\fG(E_K^n\, E_J^n)=\rho_K\,\rho_J$ and $\fG(\mathrm{id}_{\Lambda\RR^n})=\mathrm{id}_{\fA_{\Lambda\RR^n}}$.

(ii) To prove \eqref{eq:product} we consider
the homomorphisms $\chi_{\lambda}$ of $\Lambda\RR^n$, $\lambda\in\RR^n$, given 
by the action $\eta_i\mapsto\lambda_i\eta_i$ on the generators. Obviously it holds that
\be\label{eq:lambda}
\chi_\la\,E_K^n=E_K^n\,\chi_\la=\la^K\,E_K^n\quad\text{with}\quad \lambda^K\doteq\prod_{k\in K}\lambda_k.
\ee
Looking at the pertinent homomorphism $(\fG\chi_\la)$ of $\fA_{\Lambda\RR^n}$ 
and using part 3 of Definition \ref{def:Grass-functor}, the formula \eqref{eq:lambda} turns into
\be
(\fG\chi_{\lambda})\,\rho_K=\rho_K\,(\fG\chi_{\lambda})=\lambda^{K}\,\rho_K \ .
\ee
Hence, we obtain
\be
(\fG\chi_{\lambda})\,\rho_{K}\bigl(\rho_J(a)\,\rho_I(b)\bigr)=
\rho_K\bigl((\fG\chi_{\lambda})\rho_J(a)\,\,(\fG\chi_{\lambda})\rho_I(b)\bigr),
\ee
which implies
%the homomorphisms $\chi_{\lambda}$ of $\fA_{\Lambda\RR^n}$, $\lambda\in\RR^n$, induced 
%by the action $\eta_i\mapsto\lambda_i\eta_i$ on the generators of $\Lambda\RR^n$. From Definition \ref{def:Grass-functor} we get
%\be
%\chi_{\lambda}\rho_K=\rho_K\chi_{\lambda}=\lambda^{K}\rho_K \ ,
%\ee
%with $\lambda^K=\prod_{k\in K}\lambda_k$, hence
%\be
%\chi_{\lambda}(a)=\sum_K\lambda^K\rho_K(a)\ .
%\ee
%But 
%\be
%\chi_{\lambda}\rho_{K}(\rho_J(a)\rho_I(b))=\rho_K(\chi_{\lambda}\rho_J(a)\chi_{\lambda}\rho_I(b))
%\ee
%hence
\be
\lambda^K\,\rho_K\bigl(\rho_J(a)\,\rho_I(b)\bigr)=\lambda^J\,\lambda^I\,\rho_K\bigl(\rho_J(a)\,\rho_I(b)\bigr)\quad\forall\la\in\RR^n.
\ee
We conclude that
\be
\rho_K\bigl(\rho_J(a)\,\rho_I(b)\bigr)=0\quad\text{unless}\quad K=I\cup J,\,\, I\cap J=\0.
\ee
Therefore, by using also \eqref{eq:completeness}, we may write
\be
\rho_K(ab) =\sum_{J,I}\rho_K\bigl(\rho_J(a)\,\rho_I(b)\bigr)=
\sum_{J\subset K}\rho_K\bigl(\rho_J(a)\,\rho_{K\setminus J}(b)\bigr).
\ee
In view of the formula \eqref{eq:projectionsrho} for $\rho_K$, we note that
%In addition, for $J\subset K$
\be
\pi_K\bigl(\rho_J(a)\,\rho_{K\setminus J}(b)\bigr)=\rho_J(a)\,\rho_{K\setminus J}(b)\quad\text{for}\quad J\subset K,
\ee
and for $K_0\subsetneq K$
\be
\pi_{K_0}\bigl(\rho_J(a)\,\rho_{K\setminus J}(b)\bigr)=\pi_{K_0}\rho_J(a)\,\,\pi_{K_0}\rho_{K\setminus J}(b)=0
\ee
since at least one of the factors vanishes. So we arrive at
\be
\rho_K\bigl(\rho_J(a)\,\rho_{K\setminus J}(b)\bigr)=\rho_J(a)\,\rho_{K\setminus J}(b)
\ee
which completes the proof of \eqref{eq:product}.
\end{proof}

\paragraph{\it Second part of the proof.} Let $\fA^n \doteq\rho_{\{1,\dots,n\}}(\fA_{\Lambda\RR^n})$ be the subspace of the highest Grassmann degree elements. We have $a\in\fA^n$ iff $\pi_{\{1,\dots,n\}\setminus\{k\}}(a)=0$ for $1\le k\le n$. 
We define products 
\be
\fA^n\times\fA^m\to\fA^{n+m}
\ee
by
\be\label{eq:ab}
a\cdot b \doteq (-1)^{m\,\mathrm{dg}(a)}\fG\chi^{n+m}_{\{m+1,\dots,m+n\}}(a)\,\fG\chi^{n+m}_{\{1,\dots,m\}}(b)
\ee
where, for $J\equiv\{j_1,\dots,j_{|J|}\}\subset\{1,\ldots,n\}$, $\chi^n_J:\Lambda\RR^{|J|}\to \Lambda\RR^n$ is the homomorphism induced by $\eta_i\mapsto\eta_{j_i}$ with $j_1<j_2<\dots<j_{|J|}$.
The term on the right hand side of the equation is indeed an element of $\fA^{n+m}$. Namely we have
\be
P_K\circ\chi^n_J=\chi^n_J\circ P_{\{i|j_i\in K\}}
\ee  
hence
\begin{multline}
    \pi_{\{1,\dots,n+m\}\setminus\{k\}}(a\cdot b)\\
    =\pm\fG\chi^{n+m}_{\{m+1,\dots,m+n\}}\circ \pi_{\{1,\dots,n\}\setminus\{k-n\}}(a)
\cdot\fG\chi^{n+m}_{\{1,\dots,m\}}\circ\pi_{\{1,\dots, m\}\setminus\{k\}}(b)=0
\end{multline}
since for $k\le n$ the second and for $k>n$ the first factor vanishes.

The product is associative. This follows from a straightforward calculation. Let $a\in\fA^n$, $b\in\fA^m$ and $c\in\fA^k$. Then
\begin{multline}
    (a\cdot b)\cdot c=(-1)^{\mathrm{dg}(a)m+\mathrm{dg}(a)k+\mathrm{dg}(b)k}\\
    \cdot\fG\chi^{n+m+k}_{\{k+m+1,\dots,k+m+n\}}(a)\,\fG\chi^{n+m+k}_{\{k+1,\dots,k+m\}}(b)\,\fG\chi^{n+m+k}_{\{1,\dots,k\}}(c)
    =a\cdot(b\cdot c)\ .
\end{multline}

In the next step we define an inductive system 
\be
\fA^n\ni a\mapsto\iota_{k,n}(a)\doteq \eta_{1}\cdots\eta_{k-n}\fG\chi^{k}_{\{k-n+1,\dots,k\}}(a)\in\fA^{k}\, \ ,\ k\ge n
\ee
with $\iota_{k,n}\circ\iota_{n,m}=\iota_{k,m}$. If $k=n\ \mathrm{mod}\ 2$\,, we can also write
\be\label{eq:induction}
\iota_{k,n}=\fG E^{kn}_{\{1,\dots,k\},\{1,\dots,n\}}
\ee
with the matrix units defined before.

This system of embeddings is compatible with the product defined above:
\begin{lemma}\label{lemma:induction}
Let $a\in\fA^m$ and $b\in\fA^k$, hence $a\cdot b\in\fA^{m+k}$. For $n\geq m$ and $l\geq k$ it then holds that
\be
\iota_{n,m}(a)\cdot\iota_{l,k}(b)=\iota_{n+l,m+k}(a\cdot b)\ .
\ee
\end{lemma}
\begin{proof}
We insert the definitions of the embeddings and the product and obtain, for the left hand side,
\begin{equation}\label{eq:i(a)i(b)}
\iota_{n,m}(a)\cdot\iota_{l,k}(b)=\epsilon\,\eta_1\dots\eta_{l-k}\eta_{l+1}\dots\eta_{l+n-m}\,
\fG\chi^{n+l}_{\{l+n-m+1,\dots,l+n\}}(a)\,\fG\chi^{n+l}_{\{l-k+1,\dots,l\}}(b)
\end{equation}
with $\epsilon=(-1)^{k(n-m+\mathrm{dg}(a))}$, and for the right hand side
\begin{equation}\label{eq:i(ab)}
\iota_{n+l,m+k}(a\cdot b)= \epsilon'\,\eta_1\dots\eta_{n+l-m-k}\, 
\fG\chi^{n+l}_{\{l+n-m+1,\dots,l+n\}}(a)\,\fG\chi^{n+l}_{\{n+l-m-k+1,\dots,n+l-m\}}(b)
\end{equation}
with $\epsilon'=(-1)^{\mathrm{dg}(a)k}$. 
Finally we use that any element of $\fA^{n+l}$ is totally antisymmetric under a permutation of the indices of the $\eta$'s,
again due to part 3 of definition \ref{def:Grass-functor}. Hence,
applying the permutation
\be
p=\begin{pmatrix}(l-k+1)&\cdots&l&(l+1)&\cdots&(l+n-m)\\
(l-k+1+n-m)&\cdots&(l+n-m)&(l+1-k)&\cdots&(l+n-m-k)\end{pmatrix}
\ee
to \eqref{eq:i(a)i(b)} we indeed obtain \eqref{eq:i(ab)}, since $\,\mathrm{sign}(p)=(-1)^{k(n-m)}$.
\end{proof}
\paragraph{\it Third part of the proof.} We use now Lemma \ref{lemma:induction} and define $\fA$ as the inductive limit of this system with injections $\iota_n:\fA^n\to\fA$ such that 
\be
\iota_k\circ\iota_{k,n}=\iota_n\quad \text{for}\quad k\geq n
\ee 
and where the product is defined by
\be\label{eq:product-A}
\iota_n(a)\cdot\iota_m(b)\doteq\iota_{n+m}(a\cdot b)\quad\text{for}\quad a\in\fA^n,\,b\in\fA^m\ .
\ee
We equip $\fA$ with a grading such that 
\be\label{eq:degree-embedding}
\mathrm{dg}(\iota_n(a))\doteq(\mathrm{dg}(a)+n)\ \mathrm{mod}\,2\ .
\ee

It remains to construct the embeddings $\sigma_G:\fA_G\to G\otimes\fA$. Again it is sufficient to consider the case $G=\Lambda\RR^n$, $n\in\NN_0$. For $J\equiv\{j_1,\ldots,j_{|J|}\}\subset\{1,\ldots,n\}$ (with $j_1<j_2<\ldots <j_{|J|}$)
let $\chi_n^J:\Lambda\RR^n\to\Lambda\RR^{|J|}$ denote the homomorphism induced by $\eta_{j_i}\mapsto\eta_i$ and $\eta_k\mapsto0$ if $k\not\in J$. (Note the relations $\chi_n^J\circ\chi_J^n=\mathrm{id}_{\Lambda\RR^{|J|}}$ and $\chi^n_J\circ\chi^J_n=P^n_J$.)
Then we define  
\be\label{eq:sigma}
\sigma_{\Lambda\RR^n}(a)\doteq\sum_{J\subset\{1,\dots,n\}}
%(-1)^{|J|\,\mathrm{dg}(a)}
\eta_J\otimes\iota_{|J|}\circ\fG\chi_n^J\circ\rho_J(a)\ .
\ee

\begin{lemma}  $\sigma_{\Lambda\RR^n}$ has the following properties:
\begin{itemize}
    \item[$(i)$] It satisfies the naturality condition \eqref{eq:natur}.
    \item[$(ii)$] It is a homomorphism of graded algebras.
\end{itemize}
\end{lemma}

\begin{proof}
$(i)$ Let $\chi$ be a homomorphism from $\Lambda\RR^n$ to $\Lambda\RR^m$. For the right hand side of \eqref{eq:natur} we obtain
\begin{align}\nonumber
\fG^\fA\chi\bigl(\sigma_{\Lambda\RR^n}(\rho_J(a))\bigr)&=
%\sum_{K\subset\{1,\dots,m\}}
%(-1)^{|J|\,\mathrm{dg}(a)}
\chi(\eta_J)\otimes\iota_{|J|}\circ\fG\chi_n^J\circ\rho_J(a)\\
&=\sum_{K\subset\{1,\dots,m\}}
%(-1)^{|J|\,\mathrm{dg}(a)}
c_{JK}\,\theta_K\otimes\iota_{|J|}\circ\fG\chi_n^J\circ\rho_J(a),\label{eq:rhs}
\end{align}
by using \eqref{eq:GA} and \eqref{eq:chi-linear}; we recall that $\chi_{JK}$ is nonvanishing only if $|K|-|J|\in \{0,2,4,\ldots\}$.
Inserting the definitions into the left hand side we get
\begin{equation}\label{eq:lhs}
\sigma_{\Lambda\RR^m}\bigl(\fG\chi(\rho_J(a))\bigr)=
\sum_{K\subset\{1,\dots,m\}}
%(-1)^{|J|\,\mathrm{dg}(\fG\chi(a))}
\theta_K\otimes\iota_{|K|}\circ\fG\chi_m^K\circ\rho_K\bigl(\fG\chi(\rho_J(a))\bigr).
%=&\sum_{J\subset\{1,\dots,m\}}
%(-1)^{|J|\,\mathrm{dg}(a)}
%\eta_J\otimes\iota_{|J|}\circ\chi_m^J\circ\rho_J(a).
\end{equation}
Both expressions are equal, namely for \eqref{eq:lhs} we use
\be
\rho_K(\fG\chi)\rho_J=\fG(E^m_K\chi E^n_J)=c_{JK}\fG E^{mn}_{KJ}
\ee
and
\be
\chi_m^K\circ E^{mn}_{KJ}=E^{|K|,n}_{\{1,\dots,|K|\},J}\ ,
\ee
So we obtain that \eqref{eq:lhs} is equal to
\be
\sum_K c_{JK}\,\theta_K\otimes\iota_{|K|}\circ\fG E^{|K|,n}_{\{1,\dots,|K|\},J}(a).
\ee
For \eqref{eq:rhs} we indeed obtain the same result, by inserting properties of the inductions $\iota$,
\be
\iota_{|J|}=\iota_{|K|}\circ\iota_{|K|,|J|},
\ee
\be
\iota_{|K|,|J|}=\fG E^{|K|,|J|}_{\{1,\dots|K|\},\{1,\dots,|J|\}}
\ee
by using that $|K|-|J|\in \{0,2,4,\ldots\}$, and finally
\be
E^{|K|,|J|}_{\{1,\dots|K|\},\{1,\dots,|J|\}}\chi_n^J E^n_J=E^{|K|,n}_{\{1,\dots,|K|\},J}\ .
\ee
$(ii)$
The degree is preserved,  $\mathrm{dg}(\sigma_{\Lambda\RR^n}(a))=\mathrm{dg}(a)$, as a consequence of \eqref{eq:degree-embedding}.
To prove that also the product is preserved we use \eqref{eq:product}
%that $\sigma_{\Lambda\RR^n}$ intertwines the products.
and find
%Explicitly, formula \eqref{eq:product} yields
\begin{align}\label{eq:s(ab)}
\sigma_{\Lambda\RR^n}(a\,b)&=\sum_{K\subset\{1,\dots,n\}}\sum_{J\subset K} 
\eta_K\otimes\iota_{|K|}\circ(\fG\chi_n^K)\bigl(\rho_J(a)\,\rho_{K\setminus J}(b)\bigr)\nonumber\\
&=\sum_{K\subset\{1,\dots,n\}}\sum_{J\subset K}\sigma_{\Lambda\RR^n}\bigl(\rho_J(a)\,\rho_{K\setminus J}(b)\bigr)\ .
\end{align}
On the other hand we have
\begin{multline}\label{eq:s(a)s(b)}
\sigma_{\Lambda\RR^n}(a)\cdot\sigma_{\Lambda\RR^n}(b)\\
=\sum_{J,L\subset\{1,\dots,n\}}
(-1)^{\mathrm{dg}(a)|L|}\,\eta_L\eta_J\otimes
\bigl(\iota_{|J|}\circ\fG\chi_n^J\circ\rho_J(a)\bigr)\cdot\bigl(\iota_{|L|}\circ\fG\chi_n^L\circ\rho_L(b)\bigr)
\end{multline}
where only disjoint pairs $L,J$ contribute, since otherwise $\eta_L\eta_J=0$.

Using \eqref{eq:product-A} and setting $K\doteq J\cup L$, we 
%see that \eqref{eq:s(a)s(b)} is equal to
have 
\be\label{eq:s(a)s(b)1}
\bigl(\iota_{|J|}\circ\fG\chi_n^J\circ\rho_J(a)\bigr)\cdot\bigl(\iota_{|L|}\circ\fG\chi_n^L\circ\rho_L(b)\bigr)=
\iota_{|K|}\bigl((\fG\chi_n^J)\circ\rho_J(a)\cdot
(\fG\chi_n^{L})\circ\rho_{L}(b)\bigr) \ ,
\ee
and by \eqref{eq:ab} we get
\be\label{eq:s(a)s(b)2}
(\fG\chi_n^J)\circ\rho_J(a)\cdot (\fG\chi_n^{L})\circ\rho_{L}(b)=(-1)^{\mathrm{dg}(a)|L|}\,\fG(\chi^{K}_{n}\circ\chi_{JL})(\rho_J(a)\rho_L(b)).
\ee
where $\chi_{JL}$ is the automorphism of $\Lambda\RR^{n}$ which is induced by a permutation $p_{JL}$ on the indices of its generators. 
$p_{JL}\in S_{n}$ maps $(l_1,\ldots,l_{|L|},j_1,\ldots,j_{|J|})$ into $(k_1,\ldots,k_{|K|})$ and acts trivially on the remaining indices. 
Here 
$K= J\cup L=\{k_1,\ldots,k_{|K|}\}$
with $k_1<\dots < k_{|K|}\,$, $J=\{j_1,\ldots,j_{|J|}\}$ with $j_1<\dots < j_{|J|}$
and $L=\{l_1,\ldots,l_{|L|}\}$ with $l_1<\dots < l_{|L|}$.

We insert \eqref{eq:s(a)s(b)2} and \eqref{eq:s(a)s(b)1} into \eqref{eq:s(a)s(b)} and obtain
\be\label{eq:s(a)s(b)3}
\sigma_{\Lambda\RR^n}(a)\cdot\sigma_{\Lambda\RR^n}(b)=\sum_{J,L\subset\{1,\dots,n\},J\cap L=\0}\eta_L\eta_J\otimes\iota_{|K|}\circ\fG(\chi_n^{K}\circ\chi_{JL})(\rho_J(a)\rho_L(b))\ .
\ee
Since 
\be
\eta_K=\chi_{JL}(\eta_L\eta_J)\ \ee
and since $\rho_K$ acts trivially on $\rho_J(a)\rho_L(b)$ and commutes with $\fG\chi_{JL}$, 
we may write \eqref{eq:s(a)s(b)3} as (notice that $L=K\setminus J$)
\begin{align}
&\sum_{K\subset\{1,\dots,n\}}\sum_{J\subset K}\chi_{JL}^{-1}(\eta_K)\otimes \iota_{|K|}\circ\fG(\chi_n^{K})\circ\rho_K\bigl((\fG\chi_{JL})(\rho_J(a)\rho_L(b))\bigr)\nonumber\\
&\quad=\sum_{K\subset\{1,\dots,n\}}\sum_{J\subset K}(\fG^\fA\chi_{JL}^{-1})\circ
\sigma_{\Lambda\RR^n}\bigl((\fG\chi_{JL})(\rho_J(a)\rho_L(b))\bigr)\ .
\end{align}
The latter expression coincides with \eqref{eq:s(ab)} by the naturality of $\sigma_{\Lambda\RR^n}$.
%the latter expr. Here we use that  $\rho_K$ acts trivially on $\rho_J(a)\rho_L(b)$ and commutes with $\fG\chi_{JL}$.
\end{proof}

\paragraph{\it Fourth part of the proof.} 
To complete the proof of the Theorem, we still have to verify the statement about the universality of $\fA$.
Let $\fA'$ be a graded algebra and $\sigma'$ a natural transformation from $\fG$ to $\fG^{\fA'}$. Taking into account that for any $a\in\fA$ there is an $n\in\NN$ such that $a=\iota_n(a_0)$ for some uniquely fixed $a_0\in\fA^n$ and that for this $n$ the definition
\eqref{eq:sigma} gives $\sigma_{\Lambda\RR^n}(a_0)=\eta_{\{1,\dots,n\}}\otimes  a$,
we define a homomorphism $\tau:\fA\to\fA'$ by
\be
%(-1)^{n\,\mathrm{dg}(a)}
\eta_{\{1,\dots,n\}}\otimes \tau(a)=\sigma_{\Lambda\RR^n}'(a_0)\ ,\,\,a_0\in\fA^n\ .
\ee
For an arbitrary $b\in\fA_{\Lambda\RR^n}$ we easily check 
\begin{align}
    (\mathrm{id}\otimes\tau)\circ\sigma_{\Lambda\RR^n}(b)
&=\sum_J(\mathrm{id}\otimes\tau)\circ\sigma_{\Lambda\RR^n}(\rho_J (b))\nonumber\\
&=\sum_J
%(-1)^{|J|\,\mathrm{dg}(\rho_J(a))}
(\mathrm{id}\otimes\tau)
(\eta_J\otimes\iota_{|J|}(\fG\chi^J_n\circ\rho_J(b)))\nonumber\\
&=\sum_J
%(-1)^{|J|\,\mathrm{dg}(\rho_J(a))}
\eta_J\otimes\tau\bigl(\iota_{|J|}(\fG\chi^J_n\circ\rho_J(b))\bigr)\\
&=\sum_J\fG^{\fA'}\chi^n_J\circ\sigma'_{\Lambda\RR^{|J|}}\circ\fG\chi_n^J(\rho_J(b))\nonumber\\
&=\sum_J\sigma'_{\Lambda\RR^n}(\rho_J(b))
=\sigma'_{\Lambda\RR^n}(b)\nonumber
\end{align}
where the second last equality follows from the naturality of $\sigma'$.\end{proof}

%\todo{prove the universality property of the construction}

\section{The algebra of Fermi fields}\label{sec:Axioms}
We choose now $V=\Gamma(M,E)$ where $M$ is a globally hyperbolic spacetime and denote by $V_c$ its subspace of compactly supported sections.  $V$ is interpreted as the space of field configurations.
%Let $\Ec\oplus\Ec^c\ni(\phi,\phi_0)\stackrel{h}{\mapsto} \phi+\phi_0\in\Ec$.
%Then $h^*A$ is called a shifted fermionic functional. 
%We associate to each local functional $F$ on $\Ec\oplus\Ec^c$ a form $S(F)$ on $\Ec^c$ with values in a $\ZZ_2$ graded unital C*-algebra $\fA$. $S$ is compatible with the grading in the following sense:
%\be
%\bar{\kappa}(S(F))_n=(-1)^nS(\kappa_0(F))_n
%\ee
%where $\kappa_0$ is the usual grading on forms. The support is defined in terms of $\Ec$, \ie in \eqref{support} one has to find the first entry in the first summand of $V=\Ec\oplus\Ec^c$.
Let $\Floc$ be the space of local 
%compactly supported 
fermionic functionals on $V$, and let $L$ denote a generalized fermionic Lagrangian on $V$, \ie\ a map $C_0^{\infty}(M)\ni f\mapsto L(f)\in\Floc$ with $\supp L(f)\subset\supp f$ and with $L(f+g+f')=L(f+g)-L(g)+L(g+f')$  if $\supp f\cap\supp f'=\0$. 
We restrict ourselves to generalized Lagrangians that lead to Green hyperbolic \cite{GreenBear}  equations of motion.

We construct a covariant Grassmann multiplication algebra  $\mathfrak{G}:\mathfrak{Grass}\to\mathfrak{Alg}^{\ZZ_2}$ in the sense of definition \ref{def:Grass-functor}.
%, the category of unital  $\ZZ_2$-graded complex associative algebras, 
%together with embeddings $\iota_G:G\to\mathfrak{G} G\equiv\fA_G$ such that
%\be
%\iota_G(\eta)a=(-1)^{\mathrm{dg}(\eta)\mathrm{dg}(a)}a\iota_G(\eta)
%\ee
%for $\eta\in G$,  $a\in\fA_G$, and
%\be
%\iota_{G'}\circ\chi=\mathfrak{G}\chi\circ\iota_G
%\ee
%for homomorphisms $\chi:G\to G'$.
The algebras $\fA_G\equiv\fG G$ are generated by invertible elements%
\footnote{$S_G(F)$ can be expanded into a finite combination of (products of) Grassmann variables; %and is invertible if 
%the coefficient of $1_{\fG}$ is invertible.}
such a combination is invertible if and only if the coefficient of $1_{G}$ is invertible.}
$S_G(F)$ with $F\in G\otimes\Floc$ with the following properties and relations:

%  and $\Gc$ a Grassmann algebra with an involution $\eta \mapsto \eta^*$ which is compatible with the grading. We associate to any even selfadjoint element $F\in\Gc\hot\Floc$ an element $S(F)$ of the even subalgebra of the twisted tensor product $\Gc\hot\fA$ of $\Gc$ with a graded C*-algebra $\fA$. 
%We require for $S$ the following properties:
\begin{itemize}
\item (Parity) $S_G(F)$ is even for even $F$.
\item (Naturality) If $\chi:G\to G'$ is a homomorphism of Grassmann algebras 
%and $\chi_*$ is the induced map on twisted tensor products with some graded algebra, 
then
\be
S_{G'}\circ \fG^{\Floc}\chi=\mathfrak{G}\chi\circ S_G\ .
\ee
%\begin{figure}[htb]
%\centering
$$
\begin{tikzpicture}
%\matrix(M)[matrix of math nodes, row sep=2em, column sep=2em]{
%& G &&&&&& G' &\\
%&&&&&&&&\\
%G\ox \Floc &&&&&& G'\ox \Floc &&\\
%&& \fA_G &&&&&& \fA_{G'}\\
%   };
%\begin{scope}[every node/.style={midway,font=\footnotesize}]
%\draw[->] (M-1-2) -- node[above] {$\chi$} (M-1-8) ;
%\draw[->] (M-1-2) -- node[above left] {$\fG^{\Floc}$} (M-3-1) ;
%\draw[->] (M-3-1) -- node[above right] {$\fG^{\Floc}\chi$} (M-3-7);
%\draw[->] (M-3-1) -- node[below left] {$S_G$} (M-4-3) ;
%\draw[->] (M-1-2) -- node[above right] {$\fG$} (M-4-3) ;
%\draw[->] (M-4-3) -- node[below] {$\fG\chi$} (M-4-9);
%\draw[->] (M-1-8) -- node[above left] {$\fG^{\Floc}$} (M-3-7);
%\draw[->] (M-3-7) -- node[below left] {$S_{G'}$} (M-4-9) ;
%\draw[->] (M-1-8) -- node[above right] {$\fG$} (M-4-9) ;
%\end{scope}
\matrix(M)[matrix of math nodes, row sep=1.5em, column sep=1.5em]{
G\ox \Floc&&&  G'\ox \Floc\\
&&&\\
\fA_G &&& \fA_{G'}\\
   };
\begin{scope}[every node/.style={midway,font=\footnotesize}]
\draw[->] (M-1-1) -- node[above] {$\fG^{\Floc}\chi$} (M-1-4);
\draw[->] (M-1-1) -- node[left] {$S_G$} (M-3-1) ;
\draw[->] (M-3-1) -- node[below] {$\fG\chi$} (M-3-4);
\draw[->] (M-1-4) -- node[right] {$S_{G'}$} (M-3-4) ;
\end{scope}
\end{tikzpicture}
$$
%\caption{xy}
%\label{fg:1}
%\end{figure}
%In the case that $G\subset G'$ we omit the symbol $\iota_*$ for the embedding $\iota$. 
\item (Quantization condition) $S_G(\eta)=\iota_G(e^{i\eta})$ for $\eta\in G$.
\item (Causal factorization) 
\be\label{eq:caus-fact}
S_G(F_1+F_2+F_3)=S_G(F_1+F_2)S_G(F_2)^{-1}S_G(F_2+F_3)
\ee
for even functionals $F_1,F_2,F_3$ with $\supp F_1\cap J_-(\supp F_3)=\0$ where $J_-$ denotes the past of the region in the argument. 
\item (Dynamics)
Let $\vec h=\sum_{i\in I} \eta_i {\vec h}^i$ with odd elements $\eta_i\in G$, $\vec h^i \in V_c$ and $I\in\Ind$.%
\footnote{At variance with the notations in \eqref{eq:shift-F}, the Grassmann algebra $G$ considered here contains the Grassmann variables 
appearing in both the unshifted argument $\exp\sum\eta_i\vr^i$ and the shift $\vect{h}$.} 
Then
\be\label{eq:dynamics} 
S_G(F)=S_G(F^{\vec h}+\delta_{\vec h}L)
\ee
where 
\be
\delta_{\vec h}L=L(f)^{\vec h}-1_G\otimes L(f)
\ee
with $f\equiv 1$ on $\supp \vec h$ and the unit $1_G$ of $G$.
\end{itemize}
Note that the Quantization condition implies $S_G(0)=1_{\fA_G}$. Setting $F=0$ in the relation Dynamics, we obtain
\be
S_G(\delta_{\vec h}L)=1_{\fA_G},
\ee
which is characteristic for the \emph{on-shell} algebra, cf.~\cite{BF19} and Sect.~\ref{sec:perturbative}.
%We expect that a generic element of $\fA_G$ is of the form
%\be\label{eq:expansion}
%A=\sum_J\eta_J a_J
%\ee
%where 
%\be
%\eta_J=\prod_{j\in J}\eta_j\ ,\ %J\subset\{1,\dots,n\}\ 
%\ee
%with the generators $\eta_1,\dots,\eta_n$ of $G$ and coefficients $a_J\in\fA$, where $\fA=\fA_0\oplus\fA_1$ is a graded algebra
%being independent of $G$. 
We apply now Theorem \ref{thm:Grass} and 
obtain a graded algebra $\fA$ and embeddings $\sigma_G:\fA_G\to G\otimes \fA$.

We still have to equip our algebras with an antilinear involution. On a real Grassmann algebra $\Lambda V$ over some real vector space $V$ we define an involution by $\vr^*=\vr$ for $\vr\in \Lambda^1(V)=V$, for linear maps $A$ from $\Lambda V$ to some graded *-algebra by 
\be\label{eq:starr} 
A^*(\omega)=(-1)^{\dg(A)\dg(\omega)}A(\omega^*)^*\ ,\ \omega\in\Lambda V
\ee
and for the tensor product $G\otimes \fA$ of a Grassmann algebra $G$ with a graded *-algebra $\fA$ we set 
\be
(\eta\otimes a)^*=(-1)^{\dg(\eta)\dg(a)}\eta^*\otimes a^*\ ,\ \eta\in G\ ,\ a\in \fA\ .
\ee 
For a covariant Grassmann multiplication algebra $\fG$ we require that the algebras $\fG G$ are *-algebras and the embeddings $\iota_G:G\to\fG G$ are *-homomorphisms. The algebras $\fA_G=\fG G$ defined by the axioms above obtain a *-operation by  $S_G(F)^*=S_G(F^*)^{-1}$.
The subspaces $\fA^n\subset \fA_{\Lambda\RR^n}$ are invariant under the *-operation. The involution on the inductive limit $\fA$ is induced by 
\begin{equation}\label{eq:star}
    \iota_n(a)^*\doteq(-1)^{n(n-1)/2+n(\dg(a)+n)}\iota_n(a^*)\ .
\end{equation}
Indeed, since for $a\in \fA^n$, $b\in\fA^m$ equation \eqref{eq:ab}
implies that
\begin{equation}
    (a\cdot b)^*=(-1)^{m\,\dg(a)+n\,\dg(b)+nm}\,b^*\cdot a^*\ ,
\end{equation}
the involution satisfies the condition
\begin{equation}
    (\iota_n(a)\iota_m(b))^*=\iota_m(b)^*\iota_n(a)^*\ .
\end{equation}
We observe that $(\sigma_G)_G$ then is a family of *-homomorphisms. Namely, let $G=\Lambda\RR^n$ and
$\fA_{\Lambda\RR^n}\ni a=\rho_J(a)$ for some $J\subset\{1,\dots,n\}$. Using that $\dg(\eta_J)=|J|$,
$\eta_J^*=(-1)^{|J|\,(|J|-1)/2}\eta_J$, $\dg(\iota_{|J|}\circ\fG\chi_n^J(a))=(\dg(a)+|J|)\,\mathrm{mod}\,2\,$
and $(\fG\chi_n^J(a))^*=\fG\chi_n^J(a^*)$, we obtain
\begin{align}
 \sigma_{\Lambda\RR^n}(a)^* & =\bigl(\eta_J\otimes(\iota_{|J|}\circ\fG\chi_n^J(a))\bigr)^*\nonumber\\
 &=
 (-1)^{|J|\,(\dg(a)+|J|)}\,
\eta_J^*\otimes(\iota_{|J|}\circ\fG\chi_n^J(a))^*\nonumber\\
&=(-1)^{|J|\,(\dg(a)+|J|)+|J|(|J|-1)/2}\,\eta_J\otimes(\iota_{|J|}\circ\fG\chi_n^J(a))^*\\
&=\eta_J\otimes(\iota_{|J|}\circ\fG\chi_n^J(a^*))\nonumber\\
&=\sigma_{\Lambda\RR^n}(a^*)\ .\nonumber
\end{align}
Hence, $\sigma_G\circ S_G$ behaves under the $*$-operation equally to $S_G$, to wit,
$\sigma_G\bigl(S_G(F)\bigr)^*=\sigma_G\bigl(S_G(F^*)\bigr)^{-1}$.
The involution on $\fA$ is universal, in the sense that the homorphism $\tau$ in Theorem \ref{thm:Grass} is a *-homomorphism provided $\sigma'$ preserves the *-structure.  
%with Here $F_G\mapsto F_G^*$ is an involution on $G\otimes \Floc$ defined by  of %$G$ and $F_n^*=(-1)^{n(n-1)/2}F_n$ for $n$-linear
%fermionic functionals. 

In the following we omit the symbols $\sigma_G$ by identifying $\fA_G$ with a subalgebra of $G\otimes \fA$. 

Note that the ideal of $G\otimes\mathfrak{A}$ generated by the generators of $G$ is annihilated by every positive linear functional on $G\otimes\mathfrak{A}$. 
%%%%%%%%%%%%%%%%
\section{Canonical anticommutation rules}\label{sec:CAR}
We specialize now to the Dirac field on Minkowski space for simplicity, the generalization to globally hyperbolic spacetimes being straigthforward (see, e.g. \cite{DHP09}).  The space of field configurations $h\in V$ is the space of smooth sections of the spinor bundle, equipped with a nondegenerate Lorentz invariant sesquilinear form $(u,\vr)\mapsto\overline{u}\vr$ on each fiber. (Note that 
$\overline{u}$ does not mean complex conjugation, see \eqref{eq:sesquilin}.) We may choose
$ V=C^{\infty}(\MM,\CC^4)$ with the $\mathrm{Spin}(2)\equiv
\mathrm{SL}(2,\CC)$ action on $\CC^4$ by the matrix representation
\be
\mathrm{SL}(2,\CC)\ni A\mapsto \left(\begin{array}{cc}
A&0\\
0&(A^*)^{-1}
\end{array}\right)
\ee
which corresponds to the choice of $\gamma$-matrices
\be
\gamma_0=\left(\begin{array}{cc}
0&1\\
1&0
\end{array}\right)
\ ,\
\gamma_i=\left(\begin{array}{cc}
0&\sigma_i\\
-\sigma_i&0
\end{array}\right)\ ,\ i=1,2,3\ .
\ee
The sesquilinear form is obtained from the standard scalar product $(\cdot,\cdot)$ on $\CC^4$ by
\be\label{eq:sesquilin}
\overline{u}\vr=(u,\gamma_0 \vr)
\ee
The $\gamma$-matrices are then hermitian with respect to the sesquilinear form.

For compactly supported sections we can define a sesquilinear form by
\be\label{eq:sesqui}
\langle h_1,h_2\rangle=\int dx\,\overline{h_1(x)}h_2(x)\ .
\ee
The classical Dirac field $\psi$ is the evaluation functional
\be
\psi(x): V\to\CC^4;\quad  \psi(x)[h]\doteq  h(x)
\ee
and the conjugate field $\ovl\psi$ maps the configuration into the dual space
\be
\ovl\psi(x)\,:\,  V\to(\CC^4)^*;\quad  \ovl\psi(x)[h_1](\vr)\doteq \ovl{h_1(x)}\vr\ .
\ee
%[or shorter $\ovl\psi(x)[h_1](v)\doteq \ovl{h_1(x)}v$ for $v\in\CC^4$??]
%\be
%\psi(x),\ovl\psi(x)\,:\,\Ec\equiv\Cc^\infty(\MM,\CC^4)\to\CC^4;\quad  \psi(x)[h]\doteq h(x),\,\,\ovl\psi(x)[h]\doteq \ovl h(x)\doteq h(x)^\dagger\,\ga^0,
%\ee
%where $\dagger$ denotes the adjoint in $\CC^{4\x 4}$.
Smeared fields are defined as usual, that is, $\psi(s)[h]\doteq\langle s,h\rangle$, where $s\in V_c$ is a test section of the spinor bundle, and $\ovl\psi(s)[h]\doteq\langle h,s\rangle$. Note that according to \eqref{eq:starr} we have $\psi(s)^*=-\ovl{\psi}(s)$.

The Dirac Lagrangian 
$L=\,\ovl\psi\wedge\dd\psi$ with the Dirac operator $\dd=i\gamma\partial-m$ associates to any compactly supported test function $f$  a 2-form $L(f)$ on $  V$, namely 
\be\label{eq:Lagrangian}
L(f)[h_1,h_2]=\langle fh_1,\dd (fh_2)\rangle-\langle fh_2,\dd (fh_1)\rangle\ .
\ee 
Note that $\dd$ is hermitian with respect to the sesquilinear form $\langle\cdot,\cdot\rangle$, hence $L(f)$ takes imaginary values. 

We want to use the (free) Dirac Lagrangian for constructing a covariant Grassmann multiplication algebra $\fG$, \ie\ the local S-matrices in Minkowski spacetime as in the previous section, and the relation Dynamics and the Causal factorization
to derive the anticommutation relations. 

To this end we need
to extend the used functionals to $G$-valued functionals by \eqref{eq:G-extension}. 
We have for $\eta\in G$, $s,h\in V_c$
\be
\psi(s)_G[h\eta]=\psi(s)[h]\eta=\langle s,h\rangle\eta
\ee
and 
\be
\ovl{\psi}(s)_G[h\eta]=\ovl{\psi}(s)[h]\eta=\langle h,s\rangle\eta
\ee
This suggests to extend the sesquilinear form $\langle\cdot,\cdot\rangle$ to a $G\otimes \CC$ valued map $\langle\cdot,\cdot\rangle_G$ on $(G\otimes V_c)\times (G\otimes V_c)$ by
\begin{equation}\label{eq:sesqui-G}
    \langle \eta h,h'\eta'\rangle_G=\eta\langle h,h'\rangle\eta'
\end{equation}
for $h,h'\in V_c$ and $\eta,\eta'\in G$.
We may also extend the fields $\psi$ and $\ovl \psi$ to test
sections $\eta_i s^i\in G\otimes V_c$ by
\begin{equation}
    \psi_G(\eta s)[h\eta']=\eta\psi(s)[h]\eta'=\langle \eta s,h\eta'\rangle_G
\end{equation}
and
\begin{equation}\label{eq:conjugation}
    \ovl{\psi}_G(\eta s)[h\eta']=\eta\ovl{\psi}(s)[h]\eta'=(-1)^{\dg(\eta)\dg(\eta')+\dg(\eta)+\dg(\eta')}\langle h\eta',\eta s\rangle_G
\end{equation}
hence
\begin{equation}
    \psi_G(\eta s)=\eta\psi_G(s)\ ,\ \ovl{\psi}_G(\eta s)=\eta\ovl{\psi}_G(s)
\end{equation}
The extended Lagrangian $L(f)_G$ (with spacetime cutoff $f$) is a quadratic form on even elements of $G\otimes V_c$. Namely, let $h=\sum h^i\eta_i$ with $h^i\in V$ and odd elements $\eta_i\in G$. Then 
\begin{equation}\label{eq:Lagr-G}
    L(f)_G[e^h]=\frac12 L(f)_G[hh]=\frac12\sum L(f)[h^i\wedge h^j]\eta_j\eta_i=\langle fh,\dd fh\rangle_G\ .
\end{equation}
The variation under a shift $\vect{h}=\sum_{i\in I} \vect{h}^i\theta_i$, with odd elements $\theta_i\in G$, $\vect{h}^i\in  V_c$
is then a sum of a linear and a constant functional, namely
\begin{equation}\label{eq:dL-G}
    \delta_{\,\vect{h}}L_G[e^h]=\delta_{\,\vect{h}}L_G[1+h]=\langle \vect{h},\dd h\rangle_G+\langle h,\dd \vect{h}\rangle_G+\langle\vect{h},\dd\vect{h}\rangle_G\ .
\end{equation}
%and
%\begin{equation}
%    \delta_{\,\vect{h}}L_G[h]=\langle \vect{h},\dd h\rangle_G+\langle h,\dd \vect{h}\rangle_G\ .
%\end{equation}
Since $\dd$ is selfadjoint with respect to $\langle\cdot,\cdot\rangle$ we have
\begin{equation}
    \langle \vect{h},\dd h\rangle_G=\langle \dd\vect{h}, h\rangle_G
\end{equation}
and hence, using \eqref{eq:conjugation}
\begin{equation}\label{eq:delta-L-1}
    \delta_{\,\vect{h}}L_G=\psi_G(\dd\vect{h})-\ovl{\psi}_G(\dd\vect{h})+\langle\vect{h},\dd\vect{h}\rangle_G\ .
\end{equation}
Let now $s\in (G\otimes V_c)_\mathrm{even}$ %, so $s=\sum_{i=1}^n \eta_i s^i$, and let
and let
\be \label{eq:Diracfield}
\fD_G(s)\doteq \psi_G(s)-\overline{\psi}_G(s)
\ee
be the smeared \emph{classical} ``doubled Dirac field'' viewed as an element in 
$(G\otimes\Floc)_\mathrm{even}$.

%the variations of the Dirac Lagrangian, considered in $G\otimes\Floc$, with a shift element $\vect{h}=\sum_{i\in I}\eta_i \vect{h}^i$, with odd elements $\eta_i\in G$, $\vect{h}^i\in  V_c$ and $I\in\Ind$, by which we can split the variation of the Lagrangian into even (subscript $0$) and odd (subscript $1$) elements as follows 
%\be\label{eq:delta-L}
%\delta_{\,\vect{h}}L=(\delta_{\,\vect{h}}L)_1+(\delta_{\,\vect{h}}L)_0 
%\ee
%with
%\be\label{eq:delta-L-1}
%(\delta_{\vect{h}}L)_1[h]=\sum_{i\in I}\eta_i \left(\langle h,\dd \vect{h}^i\rangle-\langle \vect{h}^i,\dd h\rangle \right )
%\ee
%and
%\be\label{eq:delta-L-0}
%(\delta_{\vect{h}}L)_0=\sum_{i,j\in I} \eta_j\eta_i \langle \vect{h}^i,\dd \vect{h}^j\rangle\ .
%\ee
%
%%A smeared Dirac field is in our formalism a $G$-valued 1-form on $\Ec$ of the form $D(\chi)$ with
%
%
%Let now $f\in (G\otimes V_c)_\mathrm{even}$, so $f=\sum_{i\in I} \eta_if^i$, and let
%\be \label{eq:Diracfield}
%\fD(f)\doteq \sum_{i\in I} \eta_i (\overline{\psi}(f^i)-\psi(f^i))=\sum_{i\in I} \eta_i \left(\langle\psi,f^i\rangle-\langle f^i,\psi\rangle\right)
%\ee
%be the smeared \emph{classical} ``doubled Dirac field'' viewed as an element in 
%%$(G\otimes\Floc)_\mathrm{even}$.

%From this definition we obtain the \emph{quantized} Dirac 
%field $\Psi$ as an algebra-valued antilinear functional on $V_c$ by
\begin{proposition}
Let  $s=\sum_{i=1}^n \eta_i s^i$ with $s^i\in V_c$ and $\eta_i$ odd elements of $G$.
The S-matrix $S_G$ built with the doubled Dirac field has the expansion
\be
S_G\bigl(\fD_G( s)\bigr)=1_{\fA}+\sum_{k=1}^n\frac{i^k}{k!}
\sum_{i_1<\dots<i_k}\eta_{i_k}\dots\eta_{i_1}B_k(s^{i_1}\wedge\dots\wedge s^{i_k})
\ee
with  $\RR$-multilinear alternating maps $B_k: V_c^k\to\fA$, $k=1,\dots,n$, (the time ordered products of the doubled Dirac field).%

\end{proposition}

\begin{proof}
Let $\chi:\Lambda\RR^n\to G$ denote the homomorphism which acts on the generators of $\Lambda\RR^n$ by $\chi(\theta_i)=\eta_i$. Then
by the naturality of $S$ we have
\begin{equation}
    S_G(\fD_G(s))=(S_G\circ\fG^{\Floc}\chi)(\fD_{\Lambda\RR^n}(\sum_{i=1}^n\theta_is^i))=(\fG\chi\circ S_{\Lambda\RR^n})(\fD_{\Lambda\RR^n}(\sum_{i=1}^n\theta_is^i))
\end{equation}
hence it suffices to treat the case $G=\Lambda\RR^n$ with generators $\eta_i,i=1,\dots,n$.
By assumption, $S_G(\fD_G(\sum_{i=1}^n \eta_is^i))$ takes values in $\Lambda\RR^n\otimes \fA$, hence it is of the form
\be
S_G\bigl(\fD_G(\sum_{i=1}^n \eta_is^i)\bigr)=\sum_{I\subset\{1,\dots,n\}}\eta_IB^I(s^1,\dots,s^n)
\ee
with $B^I(s^1,\dots,s^n)\in \fA$. 

Let $\chi_\lambda$, $\lambda\in\RR^n$ denote the homomorphism of $\Lambda\RR^n$ induced by $\eta_i\mapsto\lambda_i\eta_i$. Then by the naturality of $S$ we get
\be
B^{I}(\lambda_1s^1,\dots\lambda_ns^n)=\lambda^IB^I(s^1,\dots,s^n)
\ee
hence $B^I$ depends only on the variables $s_i,\,i\in I$ and is homogeneous of degree 1 in every entry.  In particular,
for $\la=0$ we obtain $B^{\0}=S_G(0)=1$. %By replacing $\chi$ by a permutation of the generators, we find that it is totally antisymmetric.
Moreover, as a function on $k=|I|$ variables, $B^I$ does not depend on the choice of $I$.
We set 
\be
\frac{i^k}{k!}\,B_k(s^k,\dots, s^1)\doteq B^{\{1,\dots,k\}}(s^1,\dots,s^k).
\ee
Replacing $\chi$ by a permutation $p\in S_n$ of the generators, we find 
\be
\sum_{i_1<\dots<i_m}\eta_{p(i_m)}\dots\eta_{p(i_1)}B_m(s^{i_1},\dots,s^{i_m})=
\sum_{i_1<\dots<i_m}\eta_{i_m}\dots\eta_{i_1}B_m(s^{p^{-1}(i_1)},\dots,s^{p^{-1}(i_m)})
\ee
for all $1\leq m\leq n$. Let $p$ be such that it acts nontrivially only on $\{1,\dots,k\}$, \ie, $p(j)=j$ for all $k<j\leq n$.
Identifying the coefficients of $\eta_k\dots\eta_1$ by using $\eta_{p(k)}\dots\eta_{p(1)}=(-1)^{\mathrm{sign}(p)}\eta_k\dots\eta_1$,
we see that $B_k$ is totally antisymmetric.

It remains to prove that is $B_k$ is additive in every entry. We have
\be
S_G\bigl(\fD_G(\sum_{i=1}^{k+1}\eta_i s^i)\bigr)=1_{\fA}+\sum_{m=1}^{k+1}\frac{i^m}{m!}\sum_{i_1<\dots<i_m}\eta_{i_m}\dots\eta_{i_1}B_m(s^{i_1}\wedge\dots\wedge s^{i_m})\ .
\ee  
We now choose the homomorphism $\chi$ which maps $\eta_{k+1}$ to $\eta_k$ and leaves all other generators invariant. 
Identifying again the coefficients of $\eta_k\dots\eta_1$, we find
\be
B_k(s^1\wedge\dots\wedge (s^{k}+s^{k+1}))=B_k(s^1\wedge\dots\wedge s^{k})+B_k(s^1\wedge\dots\wedge s^{k-1}\wedge s^{k+1})\ .
\ee
\end{proof}

We now use $f=\eta s$ as the smearing object for $\fD$, with $s\in V_c$ and $\eta$ a generator of $G$. The involution on $\fA_G$ is defined by $S_G(\fD_G(\eta s))^*=S_G(\fD_G(\eta s)^*)^{-1}$, and $\fD(s)=\psi(s)-\ovl\psi(s)$
is selfadjoint.  The above Proposition implies
\begin{equation}
    S_G(\fD_G(\eta s))^*=1-i B_1(s)^*\eta
\end{equation}
and
\begin{equation}
    S_G(\fD_G(\eta s)^*)^{-1}=S_G(\fD_G(-\eta s))^{-1}=(1-i\eta B_1(s))^{-1}=1+i\eta B_1(s)
\end{equation}
Since $B_1(s)$ anticommutes with $\eta$, it is selfadjoint. We decompose it in its complex linear and antilinear parts,

%$\eta^2=0$, we see that $S(\fD(\eta s))\in G\ox\fA$ is of the form $S(\fD(\eta s))=B_0+\eta\, B_1$, for some $B_0,B_1\in\fA$,
%where $B_0$ is even and $B_1$ is odd.
%(To simplify the notations in this section, we write $S$ for $\sigma_G\circ S_G$ and $\eta\,A$ for $\eta\ox A\in G\ox\fA$.)
%From the naturality of $\sigma_G$ and $S_G$ it follows that $B_0=S(0)=1$, by choosing a homomorphism $\chi:G\to G$ given by 
%$\eta_i\mapsto 0$ and $1\mapsto 1$. Explicitly:
%$$
%B_0=\fG^\fA\chi\bigl(S(\fD(\eta s))\bigr)=\sigma_G\circ\fG\chi\circ S_G(\fD(\eta s))=S\bigl(\fG^{\Floc}\chi(\fD(\eta s))\bigr)=S(0).
%$$
%From $S(\fD(\eta s))^*=S(\fD(\eta s)^*)^{-1}$ we obtain that $B_1^*=B_1$. 
%In addition we expect from perturbation theory
%that $\eta\, B_1$ has similar properties as $i\fD(\eta s)$ \eqref{eq:Diracfield}: it should be the sum of a linear and an 
%antilinear  functional on $V_c$, the antilinear part being linear in $s$ and vice versa.
%These results and thoughts justify the following ansatz:
\be\label{eq:QuantizedDiracField}
B_1(s)=\Psi(s)^*+\Psi(s)\ , \quad \Psi(s)\in\fA\ .% \quad s\in V_c,\,\,\,\eta\in G\ \text{odd},%\quad B,\Psi(s)\in\fA,
\ee

We interpret $\Psi$ as the \emph{quantized} Dirac field; it is an $\fA$-valued \emph{antilinear} functional on $V_c$.
%and $\Psi(s)$ is \emph{antilinear} in $s$. 

%\be\label{eq:Diracfield}
%D(\chi)(\psi)= \sum \eta_i(\langle\psi,\chi^i\rangle-\langle\chi^i,\psi\rangle)\ .
%\ee
\begin{theorem}\label{theo:CAR}
The quantized Dirac field $\Psi$ satisfies the canonical anticommutation rules over $V_c$:
\be\label{eq:CAR}
\{\Psi(s^1)^*,\Psi(s^2)^*\}=\{\Psi(s^1),\Psi(s^2)\}=0\ ,\ \{\Psi(s^1),\Psi(s^2)^*\}=\langle s^2,i\S s^1\rangle 1_{\fA}\ ,
\ee
where 
\be\label{eq:S-Delta}
\S=(i\gamma\partial+m)\Delta
\ee
with $\Delta$ the commutator function of the scalar theory.%
\footnote{Instead of the usual notation $S,S^R,S^\pm,S^F$ for the propagators of the Dirac field, we write
$\S,\S^R,\S^\pm,\S^F$, because the letter '$S$' is reserved for the $S$-matrices. With regard to the factors $(-1),i$ and $2\pi$
in the definition of these propagators, we use the conventions given in \cite[App.~A.2]{Due19}.}
\end{theorem}

\begin{proof} 
Let $f=\sum_{i\in I}\eta_i f^i$ and $g=\sum_{i\in I} \theta_i g^i$, with $f^i, g^i\in  V_c$, $\eta_i,\theta_i$ odd elements of $G$ and $I\in\Ind$. We decompose $f=f'+\dd \vect{h}$ with $\supp \vect{h}$, $\supp f'$ compact such that
$\supp f'$ does not intersect the past of $\supp g$. We may choose
\be\label{eq:decomposition}
\vect{h}=a {\S}^R f
\ee
where $a$ is a smooth function with $a\equiv1$ on a neighborhood of the past of $\supp g$, and where ${\S}^R$ denotes the retarded
inverse of $\dd$.
From \eqref{eq:delta-L-1} we have 
\be
\fD_G(\dd\vect{h})=(\delta_{\vect{h}}L)-\langle\vect{h},\dd\vect{h}\rangle_G\ ,
\ee
hence, according to the relation Dynamics, we find
\be
\begin{split}
S_G(\fD_G(f))&=S_G(\fD_G(f')+\delta_{\vect{h}}L-\langle\vect{h},\dd\vect{h}\rangle_G)\\&=S_G(\fD_G(f')^{-\vect{h}}-\langle\vect{h},\dd\vect{h}\rangle_G)\\
&=S_G\left(\fD_G(f')-\langle \vect{h},f'\rangle_G
-\langle f',\vect{h}\rangle_G-\langle\vect{h},\dd\vect{h}\rangle_G\right)\ .
\end{split}
\ee
From Causal factorization we thus obtain
\be
S_G(\fD_G(f))S_G(\fD_G(g))=S_G(\fD_G(f'+g)-\langle \vect{h},f'\rangle_G-\langle f', \vect{h}\rangle_G-\langle\vect{h},\dd\vect{h}\rangle_G)\ .
\ee
Using $f'=f-\dd\vect{h}$ we get
\be
\fD_G(f'+g)=\fD_G(f+g)+(\delta_{-\vect{h}}L)-\langle\vect{h},\dd\vect{h}\rangle_G\ .
\ee 
We now use again the relation Dynamics:
\be
S_G\bigl(\fD_G(f)\bigr)S_G\bigl(\fD_G(g)\bigr)=S_G\bigl(\fD_G(f+g)+(\delta_{-\vect{h}}L)+c\bigr)=S_G\bigl(\fD_G(f+g)^{\vect{h}}+c\bigr),
\ee
where $c\doteq -\langle \vect{h},f'\rangle_G-\langle f', \vect{h}\rangle_G-2\langle\vect{h},\dd\vect{h}\rangle_G$.
%\bigl(\fD(f')^{-\vect{h}}-\fD(f')\bigr)-2(\delta_{\vect{h}}L)_0\in G$ which satisfies $c^{\vect{h}}=c$. 
Taking into account that
\begin{align}
 \fD_G(f+g)^{\vect{h}}-\fD_G(f+g)+c & =\langle (f+g),\vect{h}\rangle_G+\langle\vect{h},(f+g)\rangle_G +c\nonumber\\
 &=\langle f-f',\vect{h}\rangle_G+\langle\vect{h},f-f'\rangle_G-2\langle\vect{h},\dd\vect{h}\rangle_G+\langle g,\vect{h}\rangle_G+\langle\vect{h},g\rangle_G\nonumber\\
 &=\langle\dd\vect{h},\vect{h}\rangle_G+\langle\vect{h},\dd\vect{h}\rangle_G-2\langle\vect{h},\dd\vect{h}\rangle_G+\langle g,\vect{h}\rangle_G+\langle\vect{h},g\rangle_G\\
 &=\langle g,\vect{h}\rangle_G+\langle\vect{h},g\rangle_G\ .\nonumber
\end{align}
%\begin{align*}
 %\fD(f+g)^{\vect{h}}-\fD(f+g) & = \bigl(\fD(f')^{\vect{h}}-\fD(f')\bigr)+\bigl(\fD(\dd\vect{h})^{\vect{h}}-\fD(\dd \vect{h})\bigr)+\bigl(\fD(g)^{\vect{h}}-\fD(g)\bigr)\\
 %& =-\bigl(\fD(f')^{-\vect{h}}-\fD(f')\bigr)+2(\delta_{\vect{h}}L)_0+\bigl(\fD(g)^{\vect{h}}-\fD(g)\bigr),
%\end{align*}
we arrive at
\be\label{Weyl}
S_G\bigl(\fD_G(f)\bigr)\,S\bigl(\fD_G(g)\bigr)=S\bigl(\fD_G(f+g)+E(f,g)\bigr)=S_G\bigl(\fD_G(f+g)\bigr)\,S_G\bigl(E(f,g)\bigr)
\ee
with $E(f,g)\in G$ given by,
\be
E(f,g)\doteq \langle g,\vect{h}\rangle_G+\langle\vect{h},g\rangle_G=
\langle g,\S^R f\rangle_G+\langle\S^R f,g\rangle_G
%\bigl(\fD(g)^{\vect{h}}-\fD(g)\bigr)
%=\sum_{i,j\in I}\theta_i\eta_j\left(\langle %\vect{h}^j,g^i\rangle-\langle g^i,\vect{h}^j\rangle\right),
\ee
where we replaced $\vect{h}=\S^R(f-f')$ by $\S^R f$ since $\supp (\S^R f')\cap\supp g=\0$. (The second equality in \eqref{Weyl} 
follows from Causal factorization and $\supp E(f,g)=\emptyset$.)

The relation \eqref{Weyl} implies the canonical anticommutation relations. 
To see this, we first observe that 
\be\label{Weyl1}
S_G(\fD_G(g))S_G(\fD_G(f))=S_G(\fD_G(f))S_G(\fD_G(g))S_G(E(g,f)-E(f,g))
\ee
with
\be
E(g,f)-E(f,g)=\langle f,\S g\rangle_G-\langle g,\S f\rangle_G
\ee
with the $G\otimes\CC$-valued sesquilinear form
\be\label{eq:S-def}
\langle g,\S f\rangle_G=\langle g,\S^R f\rangle_G-\langle\S^R g, f\rangle_G\ .
%\langle \rho^i,\S\chi^j\rangle=\langle\rho^i,\dd^R\chi^j\rangle+\langle\dd^R\chi^j,\rho^i\ranle+\langle\chi^j,\dd^R\rho^i\rangle+\langle\dd^R\rho^i,\chi^j\rangle\ .
\ee
Let now $f=\eta_1 s^1$ and $g=\eta_2 s^2$ with $s^1,s^2\in  V_c$ and odd elements $\eta_1,\eta_2\in G$. 
%Then we have the expansions
%\be
%S(\fD(f))=A_0+\eta_1 A_1\ ,\ S(\fD(g))=B_0+\eta_2 B_1\ ,
%\ee
%with $A_i,B_i\in\fA$ and, by the Quantization condition, 
Inserting
\be
S_G(E(g,f)-E(f,g))=1+i\eta_2\eta_1\left(\langle f^1,\S g^2\rangle_G +\langle g^2,\S f^1\rangle_G\right)\ .
\ee
and \eqref{eq:QuantizedDiracField} into \eqref{Weyl1}, we get a non-trivial identity only for the coefficients of $\eta_1\eta_2$:
%with $\langle\chi_1,\rho_2\rangle\in\CC$. 
%We insert this into \eqref{Weyl1} and find
%\be
%A_0B_0=B_0A_0
%\ee
%\be
%A_0B_1=B_1A_0\ ,\ A_1B_0=B_0A_1\ .
%\ee
\begin{multline}
-\bigl(\Psi(s^2)^*+\Psi(s^2)\bigr)\bigl(\Psi(s^1)^*+\Psi(s^1)\bigr)\\=\bigl(\Psi(s^1)^*+\Psi(s^1)\bigr)\bigl(\Psi(s^2)^*+\Psi(s^2)\bigr)-
i\left(\langle s^1,\S s^2\rangle_G+\langle s^2,\S s^1\rangle_G\right).
\end{multline}
%-B_1A_1=A_1B_1+A_0B_0(- i)\left(\langle f^1,\S g^2\rangle+\langle g^2,\S f^1\rangle\right)
%\ee
%We now conclude from the naturality of $S$ that $A_0=B_0=S(0)=1$ (we choose a homomorphism with $\eta_i\mapsto 0$, $1\mapsto1$).
This equation must hold individually for the terms being linear/antilinear in $s^1$ and linear/antilinear in $s^2$. Hence,
%we arrive at the anticommutation relation
%\be\label{eq:AR}
%\{B_1,A_1\}=i\left(\langle f^1,\S g^2\rangle+\langle g^2,\S f^1\rangle\right)\ .
%\ee
%From \eqref{eq:QuantizedDiracField} and \eqref{eq:AR} we get the 
we obtain the canonical anticommutation relations \eqref{eq:CAR}.
%\be\label{eq:CAR}
%\{\Psi(s^1),\Psi(s^2)\}=0\ ,\ \{\Psi(s^1),\Psi(s^2)^*\}=-i\langle s^2,\S s^1\rangle\ .
%\ee

To see that the definition $\S\doteq \S^R-(\S^R)^*$ \eqref{eq:S-def} (where $(\S^R)^*$ denotes the adjoint of $\S^R$ 
with respect to the sesquilinear form $\langle\cdot, \cdot\rangle$, which coincides with the advanced inverse of the 
Dirac operator) agrees with the explicit formula \eqref{eq:S-Delta} for $\S$, note that
\be
\S^R=(i\gamma\partial+m)\Delta^R,\quad (\S^R)^*(x)=(i\gamma\partial_x+m)\Delta^R(-x)
\ee
and $\Delta(x)=\Delta^R(x)-\Delta^R(-x)$.
\end{proof}

\begin{remark} To verify the consistency of our conventions,
we check that $\langle\cdot, i\S\cdot\rangle$ is a positive semidefinite sesquilinear form on $V_c$. 
%The retarded inverse of the Dirac operator is \be
%\S^R=(i\gamma\partial+m)\Delta^R
%\ee
%with the retarded propagator $\Delta^R$ of the scalar field. But then 
%\be
%i\S=-(\gamma\partial-im)\Delta
%\ee
%with the commutator function of the scalar theory,
%hence\todo{$+m\gamma^0$ instead of $-m$ (MD)}
From \eqref{eq:S-Delta} we obtain
\be
\gamma^0i\S(x-y)=(2\pi)^{-3}\int d^4p\,\, \delta(p^2-m^2)\epsilon(p_0)(p_0+\vec{\alpha}\cdot\vec{p}+m\gamma^0)e^{-ip(x-y)}
\ee
(where $\alpha_k\doteq\gamma_0\gamma_k$ for $k=1,2,3$) and thus
\be
\langle f,i\S f\rangle=2\pi\int d^4p\,\,\delta(p^2-m^2)\epsilon(p_0)
\left(\tilde{f}(p),(p_0+\vec{\alpha}\cdot\vec{p}+m\gamma^0)\tilde{f}(p)\right)
\ee
where $\tilde{f}$ denotes the Fourier transform of $f$.
The positivity follows now from the fact that the matrix valued function
\be
p\mapsto \epsilon(p_0)(p_0+\vec{\alpha}\cdot\vec{p}+m\gamma^0)
\ee
is positive semidefinite on both components of the mass hyperboloid $p^2=m^2$.
\end{remark}

%%%%%%%%%%%%%%%%%%%%%%%%%%%%%%%
\section{C*-structure}\label{sec:C-star}
The axioms define a graded unital *-algebra $\fA=\fA_0\oplus\fA_1$.
%\todo{Exapand on this. Get $\mathfrak{A}$ from $G\otimes \mathfrak{A}$ when mapping $G$ to 1 and using naturality.} 
We now want to equip it with a C*-norm. We start with S-matrices $S(F)$ with even fermionic functionals $F$ without auxiliary Grassmann variables. There we can proceed as in the case of a bosonic field. We look at the group generated by these elements modulo the relations Causality and the Quantization condition $S(c)=e^{ic}1$ for constant functionals $c$ and define a state on the group algebra by
\be
\omega(U)=0 \text{ for }U\not\in\{e^{ic}1|c\in\RR\}\ .
\ee  
The operator norm in the induced GNS representation is a C*-norm. We then equip the algebra with the maximal C*-norm \cite{PalmerI, PalmerII}. 
Note that in contrast to the bosonic case the Dynamical relation does not lead to relations within this algebra. 

We now want to extend this C*-norm. We cannot expect that it can be extended to the full algebra, since the presence of the Grassmann variables induces an expansion of the S-matrices into polynomials of Grassmann variables whose coefficients cannot be expected to be bounded, in general. An example is
\be\label{eq:current}
S(\eta j_\mu(f^\mu))=1+i\eta J_\mu(f^\mu),
\ee 
where $\eta$ is an even element of $G$ with $\eta^2=0$, with the classical current
\be
j_\mu(f^\mu)=\int\overline{\psi}\wedge\gamma_\mu\psi f^\mu
\ee
of the Dirac field and its quantized version $J_\mu$ (defined by \eqref{eq:current}).

Instead we use the anticommutation relations \eqref{eq:CAR} which imply that for $||f||_{V_c}=1$,
with the seminorm
\be
||f||_{V_c}^2=\langle f,i\S f\rangle\ ,
\ee
$\Psi(f)^*\Psi(f)$ is a selfadjoint projection. Hence for every non-zero C*-seminorm
\be
||\Psi(f)||=||f||_{V_c} 
\ee
holds. Moreover, we have 
\begin{proposition}
$\Psi(f)=0$ if $||f||_{V_c} =0$.
\end{proposition}
\begin{proof}
    Let $f\in V_c$ with $||f||_{V_c} =0$. Then, due to the positive semidefiniteness of $\langle\cdot, i\S\cdot\rangle$,
    we may use the Cauchy-Schwarz inequality to obtain $\langle g, i\S f\rangle=0$ for every $g\in V_c$. Thus $\S f=0$. But then, due to the general properties of normal hyperbolic operators, $f$ must be of the form $\dd h$ for some $h\in V_c$. So for an odd element $\eta\in G$ we get $\fD_G(\eta f)=\delta_{\eta h}L$, 
    hence by the axiom Dynamics
    $S_G(\fD_G(\eta f))=1$, thus $\Psi(f)+\Psi(f)^*=0$. Since $||if||_{V_c}=||f||_{V_c}$, we can repeat the argument with $if$ instead of $f$ and arrive at $\Psi(f)=0$.
\end{proof}
%(Note that $\Psi(f)=0$ if $||f||_{V_c} =0$ due to the Dynamical axiom).
We conclude that the *-algebra generated by $\Psi(f)$, $f\in V_c$ is the algebra of canonical anticommutation relations. 

Let us consider the sub-*-algebra $\fB$ of $\fA$, generated by the S-matrices $S(F)$ with even $F$ as above and the Dirac fields $\Psi(f)$, then we have
\begin{theorem}
The maximal C*-seminorm on $\fB$ exists and is a C*-norm.
\end{theorem}
\begin{proof}
Let us equip $\fB$ with the norm
\be\label{eq:banachnorm}
||A||_1=\inf\{\sum_i\prod_j||C_i^j||\ \big\vert\ A=\sum_i\prod_j U_i^jC_i^j\}
\ee
with products $U_i^j$ of S-matrices $S(F)$ and their inverses and $C_i^j$ elements of the *-algebra generated by the Dirac field, equipped with its unique C*-norm $||\cdot||$. For every element $A$ as in \eqref{eq:banachnorm} and any C*-seminorm $p$ we get $p(A)\leq \sum_i\prod_j p(C^j_i)$, since unitary elements are bounded by $1$ in every C*-seminorm, then by uniqueness of the C*-norm $||\cdot||$  one gets $p(C_i^j)=||C^j_i||$ hence the Banach norm $||\cdot||_1$ dominates every C*-seminorm, and we can equip $\fB$ with its maximal C*-seminorm.
%
%We then consider the maximal C*-seminorm\todo{Cite something for the definition of the maximal C*-norm} on the sub-*-algebra $\fB$ of $\fA$ which is generated by the S-matrices $S(F)$ with even $F$ as above and the Dirac fields $\Psi(h)$.
It remains to show that this is actually a norm. 

For this purpose we choose a family of unitaries in the algebra generated by the Dirac field which is a basis of a dense subset and which is closed under multiplication and adjunction, up to a factor. To obtain this basis, we use the fact that the algebra of canonical anticommutation relations (the CAR algebra) is, by the Jordan-Wigner transformation, isomorphic to a tensor product of $2\times 2$-matrix algebras (see for instance  \cite{BR2}, and the detailed treatment in \cite{CrDuFid}) where the unit together with the Pauli matrices form such a basis. We consider the generated group $\mathcal{U}$, together with a nontrivial (hence faithful) representation $\sigma$ of the CAR algebra. 

We then construct the induced representation of the full group $\mathcal{V}$ generated by $\mathcal{U}$ and the S-matrices $S(F)$ as above,
by proceeding as follows:
we choose from every coset $j\in \mathcal{V}/\mathcal{U}$ a representative $V_j$. Then the induced representation $\pi$ is defined on the Hilbert space
\be
\mathcal{H}_\pi=\bigoplus_{j\in\mathcal{V}/\mathcal{U}}\mathcal{H}_{\sigma}^j
\ee
where each summand is a copy of the representation space of $\sigma$, by
\be
(\pi(V)v)_i=\sum_{j\in\mathcal{V}/\mathcal{U},V_i^{-1}VV_j\in\mathcal{U}}\sigma(V_i^{-1}VV_j)v_j\quad\text{for}\quad
v=\bigoplus_{j\in\mathcal{V}/\mathcal{U}}v_j\in\mathcal{H}_\pi\ .
\ee
(Note that the sum contains only one term.)

$\pi$ can now be extended to the group algebra over $\mathcal{V}$  which is a $||\cdot||_1$-dense subalgebra $\fB_0$ of $\fB$. This representation is faithful. To see this we apply a generic element $\sum_{V\in\mathcal{V}} \lambda_V V$, with $\lambda_V\not=0$ for a finite linearly independent subset of $\mathcal{V}$, to the subspace corresponding to the coset of unity, denoted by $\hat{0}$. Assume $\sum_{V\in\mathcal{V}}\lambda_V \pi(V)=0$. Let $V_{\hat{0}}=1$ and $v\in \mathcal{H}_{\sigma}^{\hat{0}}$. Then 
\be
\sum_{V\in\mathcal{V}}\lambda_V \pi(V)v=\bigoplus_{j\in\mathcal{V}/\mathcal{U}}\sum_{V\in j}\lambda_V\sigma(V_j^{-1}V)v\ .
\ee
Since $\sigma$ has a faithful extension to the CAR algebra, we
conclude that for all $j\in\mathcal{V}/\mathcal{U}$
\be
\sum_{V\in j}\lambda_V V_j^{-1}V=0 \ .
\ee
But the set $\{V_j^{-1}V|V\in j,\lambda_V\not=0\}$ is linearly independent, hence $\lambda_V=0$ for all  $V\in j$. 
Thus the operator norm in this representation is a C*-norm on $\fB_0$. Moreover, since $\pi$ is continuous, it has a unique extension to a C*-seminorm on $\fB$. 

Finally, given any element of $B\in\fB$, $B\not= 0$,  there is some choice of $\mathcal{U}$
such that $B\in\fB_0$, hence there is a C*-seminorm nonvanishing on $B$. Thus the maximal C*-seminorm on $\fB$ is indeed a norm.
\end{proof}
\begin{remark}
The proof uses only dense *-algebras. By completing $\fB$, it is clear that the CAR C*-algebra is properly contained in it.

Moreover, by restriction to open bounded subregions of Minkowski spacetime we can define a net of C*-algebras from $\fB$.
This construction uses the Lagrangian of the free theory. But as shown in \cite{BF19} (see also \cite{BDFR21} for an explicit formula) the net of interacting observables can be constructed within the net of the free theory and vice versa. In particular operators satisfying the CAR can also be found in the interacting theory.
\end{remark}
%%%%%%%%%%%%%%%%%%%%%%%%%%%%%%%

\section{Equivalence of the relation Dynamics to the Field equation in perturbation theory}\label{sec:Pert}

For the perturbative description of Dirac spinor fields (see e.g.~\cite[Chap.~5.1.1]{Due19}), 
we aim to prove the equivalence of the relation `Dynamics' 
\eqref{eq:dynamics} to the axiom `Field equation' for time-ordered products.  
The main idea of proof is taken from \cite[Appendix]{BF19}. We throughout work with \emph{extended} fermionic functionals $F_G$
with compact support,
\ie, $F_G$ is defined on $G\ox\Lambda V$, where $V=C^\infty(\MM,\CC^4)$; this is not necessary, however, proceeding this way 
we may directly borrow some formulas from Section \ref{sec:CAR}. By $\Fc$ we mean the space of all functionals of this kind
and by $\Floc$ the subspace of the local ones.

\paragraph{Star product and unrenormalized time-ordered product.} \label{sec:perturbative}
%The basic fields are the evaluation functionals
%$$
%\psi(x),\ovl\psi(x)\,:\,\Ec\equiv\Cc^\infty(\MM,\CC^4)\to\CC^4;\quad  %\psi(x)[h]\doteq h(x),\,\,\ovl\psi(x)[h]\doteq \ovl h(x)\doteq %h(x)^\dagger\,\ga^0,
%$$
%where $\dagger$ denotes the adjoint in $\CC^{4\x 4}$.
%Let $\Fc$ denote the space of compactly supported  fermionic  functionals %on $\Ec$ and let $F_G=(\eta\ox  F)\in (G\ox\Fc)$ as above. We recall that %the 
%\emph{classical product}  $ F_{G1}\wedge  F_{G2}\doteq  (\eta_1\eta_2)\ox % F_1\wedge  F_2$ (with $\eta_1,\eta_2\in G$) is 
%associative and commutative. Exponentials of type
%\be\label{eq:exp}
%\exp_\wedge(F_G)= e_\wedge^ {F_G}\doteq 1+\sum_{n=1}^\infty\frac{1}{n!}\, %F_G^{\wedge\,n}\quad\text{with}\quad  F_G\doteq \sum_j\la_j %F_{Gj},\,\,\,\la_j\in\RR
%\ee
%will be used.
To define the \emph{star product} let
\be
\S^\pm(x)\doteq \pm(i\pp_x+m)\Dl^+(\pm x),
\ee
where $\Dl^+$ is the scalar Wightman $2$-point function (or a Hadamard function). Note that $\S^+$ and $\S^-$ are related to the "anticommutator function" $\S$ appearing in Sect.~\ref{sec:CAR} by
$\S(x)=-i\bigl(\S^+(x)+\S^-(x)\bigr)$. The star product is defined by
\be\label{eq:star-tilde}
 (\eta_1\ox F_{1,G})\star (\eta_2\ox F_{2,G})\doteq (-1)^{\dg(\eta_2)\dg(F_1)}\,(\eta_1\eta_2)\ox (F_{1,G}\star F_{2,G})
 \ee
for $F_{1,G},F_{2,G}\in\Fc$, and  $\eta_1,\eta_2\in G$, and%
%\footnote{By $\frac{\dl F_G}{\dl\psi(x)_G}(h)$ or $\frac{\dl F_G}{\dl\ovl\psi(x)_G}(h)$ we mean the integral kernel of the left 
%(functional) derivative $F_G^{(1)}(h)$ of $F_G$ at $h$ as introduced in definition \ref{GradedDerivaitve}; the functional 
%derivative from the right-hand side is defined by
%$$
%\frac{\dl_r}{\dl\psi(y)_G}\,\ovl\psi(x_1)_G\wedge\cdots\wedge\psi(x_n)_G\doteq
%(-1)^{n-1}\,\frac{\dl}{\dl\psi(y)_G}\,\ovl\psi(x_1)_G\wedge\cdots\wedge\psi(x_n)_G\,
%$$
%and similarly for $\dl_r/\dl\ovl\psi(y)_G$.}
\begin{align}\label{eq:de F-star}
F_{1,G}\star&  F_{2,G}\\
\doteq& \sum_{n,k=0}^\infty \frac{\hslash^{n+k}}{n!k!} 
\int dx_1 \cdots dx_{n+k}\,dy_1 \cdots dy_{n+k}\, 
\frac{\dl_r^{n+k}  F_{1,G}}{\dl\psi_{t_1}(x_1)_G\cdots (n)\,
\dl\ovl\psi_{u_1}(x_{n+1})_G\cdots (k)}
\nonumber\\
 &\wedge\, \prod_{j=1}^n \S^+_{t_js_j}(x_j-y_j)\, \prod_{l=1}^k
\S^-_{v_lu_l}(y_{n+l}-x_{n+l})\,\frac{\dl^{n+k}  F_{2,G}}{\dl\ovl\psi_{s_1}(y_1)_G\cdots (n)\,
\dl\psi_{v_1}(y_{n+1})_G\cdots (k)} \ ,\nonumber
\end{align}
where $\dl^n/\dl\psi_{t_1}(x_1)_G\cdots (n)\doteq \dl^n/\dl\psi_{t_1}(x_1)_G\cdots \dl\psi_{t_n}(x_n)_G$ 
and $\frac{\dl_r}{\dl\psi_G}$ denotes the functional derivative from the right-hand side.%
\footnote{By $\frac{\dl F_G}{\dl\psi(x)_G}(h)$ or $\frac{\dl F_G}{\dl\ovl\psi(x)_G}(h)$ we mean the integral kernel of the left 
(functional) derivative $F_G^{(1)}(h)$ of $F_G$ at $h$ as introduced in definition \ref{GradedDerivaitve}; the functional 
derivative from the right-hand side is defined by
$$
\frac{\dl_r}{\dl\psi(y)_G}\,\ovl\psi(x_1)_G\wedge\cdots\wedge\psi(x_n)_G\doteq
(-1)^{n-1}\,\frac{\dl}{\dl\psi(y)_G}\,\ovl\psi(x_1)_G\wedge\cdots\wedge\psi(x_n)_G\,
$$
and similarly for $\dl_r/\dl\ovl\psi(y)_G$.}

 In addition we introduce the \emph{unrenormalized time-ordered product} $\star_F$, by the same  formulas
 \eqref{eq:star-tilde} and \eqref{eq:de F-star}, but with both $\S^+(z)$ and $(-\S^-(z))$ replaced by 
 \be
 \S^F(z)\doteq (i\pp_x+m)\Dl^F(z)=\theta(z^0)\,\S^+(z)-\theta(-z^0)\,\S^-(z)
 \ee 
 everywhere, where $\Dl^F$ is the scalar Feynman propagator. 
 This product exists if the pertinent contractions do not form any loop diagram -- we shall use it only in such instances. 
 For example, for $j^\mu(x)_G\doteq \ovl\psi(x)_G\wedge\gamma^\mu\psi(x)_G$ (\ie, the electromagnetic current)  the last term in
\begin{align}
j^\mu(x)_G & \star_F j^\nu(y)_G=j^\mu(x)_G\wedge j^\nu(y)_G
+\hslash\,\ovl\psi(x)_G\wedge \ga^\mu\,\S^F(x-y)\,\ga^\nu\,\psi(y)_G\nonumber\\
&+\hslash\,\ovl\psi(y)_G\wedge \ga^\nu\,\S^F(y-x)\,\ga^\mu\,\psi(x)_G
-\hslash^2\,\mathrm{tr}\bigl(\ga^\mu\,\S^F(x-y)\,\ga^\nu\,\S^F(y-x)\bigr)\ ,
\end{align}
(where matrix notation for the spinors is used and $\mathrm{tr}(\cdot)$ 
denotes the trace in $\CC^{4\x 4}$) does generally not exist, but it is well-defined when smeared out with a test function $f(x,y)$ 
which has support outside of the diagonal $x=y$.

Both $\star$ and $\star_F$ are associative and the latter is commutative if both factors are even elements of  $G\ox\Fc$.
For $F_G=\sum_j\eta_j\ox F_{j,G}\in (G\ox\Fc)_\mathrm{even}$, exponentials $\exp_\wedge(F_G)$ and 
$\exp_{\star_ F}(F_G)$ are 
defined by the pertinent power series, where the powers are meant with respect to the indicated product.

We work with the sesquilinear form $\langle\cdot\,,\,\cdot\rangle_G$ \eqref{eq:sesqui-G} and the 
Lagrangian $L(f)_G$ \eqref{eq:Lagr-G}. We use that the variation of $L(f)_G$ under a shift, $\dl_{\vect{h}}L_G$ \eqref{eq:dL-G}
(where $\vect{h}=\sum_{i\in I} \vect{h}^i\eta_i$, with odd elements $\eta_i\in G$, $\vect{h}^i\in  V_c$),
may be written in terms 
%$$
%L(g)=\int dx\,\,g(x)\ovl\psi(x)\wedge\dd_x(g(x)\psi(x))=
%L(g)^*,\quad\text{with}\quad\dd_x\doteq i\pp_x-m
%$$
%the Dirac operator and $g\in C_0^\infty(\MM,\RR)$. We recall that $\dd$ is hermitean w.r.t.~the sesquilinear form \eqref{eq:sesqui}.
%And we also recall the definition 
of the smeared classical double Dirac field $\fD_G(s)$ \eqref{eq:Diracfield} as
%can be written as
%\be\label{eq:D(f)}
%\fD(f)=\sum_j\eta_j (\ovl\psi(f_j)-\psi(f_j))=\sum_j\eta_j (\langle\psi,f_j\rangle-\langle f_j,\psi\rangle).
%\ee
%Let $\vect{h}(x)\doteq \sum_{j\in I}\eta_j \vect{h}^j(x)$ (with $\eta_j$ odd elements of $G$, $\vect{h}^j\in  V_c$ and $I\in\Ind$) be %the field shift.  For $g\vert_{\supp\vect{h}}=1$ we rewrite \eqref{eq:delta-L-1} in the form
\be\label{eq:dL0}
\dl _{\vect{h}} L_G %\doteq  L(g)^{\vect{h}}-L(g)
=\fD_G(\dd\vect{h})+\langle\vect{h},\dd\vect{h}\rangle_G\in (G\ox\Floc)_\mathrm{even},
\ee
by using \eqref{eq:delta-L-1}.

For $F_G\in (G\ox\Fc)_\mathrm{even}$ we introduce the Euler derivative
\be\label{eq:epsF}
(\eps  F_G)(\vect{h})\doteq \frac{d}{du}\Big\vert_{u=0} F_G^{u\vect{h}}=
\int dx\,\,\Bigl(\ovl{\vect{h}(x)}\,\frac{\dl F_G}{\dl\ovl\psi(x)_G}+\frac{\dl_r F_G}{\dl\psi(x)_G}\,\vect{h}(x)\Bigr).
\ee

By using $\dd \S^+=0=\dd \S^-$ we obtain
\be\label{eq:star-K} 
 F_G\star \fD_G(\dd\vect{h})= F_G\wedge \fD_G(\dd\vect{h})=\fD_G(\dd\vect{h})\wedge F_G=\fD_G(\dd\vect{h})\star  F_G
 %\quad e_\star^{\fD(\dd\chi)}=e_\wedge^{\fD(\dd\chi)}
\ee
for all $F_G\in (G\ox\Fc)_\mathrm{even}$ and all $\vect{h}$ of the above given form. For the product $\star_F$, the relation $\dd \S^F=i\dl$ yields
\be
\exp_{\star_F}(i\,\fD_G(\dd\vect{h})) = \exp_\wedge(i\, \fD_G(\dd\vect{h}))\cdot \exp(-i\,\langle \vect{h},\dd\vect{h}\rangle_G).
\label{eq: F-1}
\ee

\paragraph{The renormalized time-ordered product.} %Let $\Floc\doteq \{F\in\Fc\,\vert\,\text{$F$ is local}\}$.
The renormalized (off-shell) time-or\-de\-red product is a collection of linear maps %$T_n\,:\,(G\ox \Floc)^{\ox n}\to (G\ox\Fc)$ or 
$T_n\,:\,\Floc^{\ox n}\to \Fc$, $n\in\NN$, which is defined by certain basic axioms and renormalization conditions, 
see e.g.~\cite[Chap.~5.1.1]{Due19}; in particular $T_{n,G}\,:\,((G\ox \Floc)_\mathrm{even})^{\ox n}\to (G\ox\Fc)_\mathrm{even}$,
defined by
\be
T_{n,G}(\eta_1\ox F_{1,G},\ldots, \eta_n\ox F_{n,G})\doteq (\eta_n\cdots\eta_2\eta_1)\ox T_n( F_{1,G},F_{2,G},\ldots, F_{n,G}),
\ee 
is required to be invariant under permutations of $(\eta_1\ox F_{1,G}),\ldots, (\eta_n\ox F_{n,G})$.
Due to the basic axiom `Causality', $T_{n,G}$ agrees with the 
$n$-fold product $\star_F$ if $\supp  F_{j,G}\cap\supp  F_{k,G}=\emptyset$  for all $j<k$. 
The generating functional of the sequence of time-ordered products $(T_{n,G})_{n\in\NN}$ is the
$S$-matrix
\be\label{eq:S-matrix}
\begin{split}\raisetag{24pt}
\qquad\qquad\qquad S_G\,:\,&(G\ox\Floc)_\mathrm{even}\to (G\ox\Fc)_\mathrm{even}\\
\shortintertext{defined by}
S_G(\la F_G)&\doteq T_G\bigl(\exp_\ox(i\la F_G)\bigr)\equiv 1+\sum_{n=1}^\infty\frac{i^n\la^n}{n!}T_{n,G}( F_G^{\ox n})\ ,
\end{split}
\ee
which we understand as a formal power series in $\la\in\RR$. In particular,
since $\dl_{\vect{h}} L_G$ does not contain any terms of second or higher order in $\psi_G,\ovl\psi_G$,
there do not contribute any loop diagrams to $S_G(\dl_{\vect{h}} L_G)$, hence we obtain
 \be\label{eq:S(dL0)}
S_G(\dl_{\vect{h}} L_G)=\exp_{\star_ F}(i\dl_{\vect{h}} L_G)=\exp_\wedge(i\fD_G(\dd\vect{h}))\ ,
\ee
where the second equality is due to \eqref{eq:dL0} and \eqref{eq: F-1}.

The renormalization condition ``(off-shell) Field equation'' can be written in terms of the retarded interacting field,
 \be
 \begin{split}
 R_G\bigl(\exp_\ox(F_G),H_G\bigr)&\doteq \frac{d}{id\la}\Big\vert_{\la=0} S_G(F_G)^{\star -1}\star S_G( F_G+\la H_G)\\
& =S_G(F_G)^{\star -1}\star T_G\bigl(\exp_\ox(i F_G)\ox H_G\bigr),
\end{split}
\ee
 (where $F_G,H_G\in (G\ox\Floc)_\mathrm{even}$ and $F_G$ is interpreted as the interaction), as
 \be\label{eq:FE-2}
 R_G\bigl(\exp_\ox(F_G),[(\eps F_G)(\vect{h})+\fD_G(\dd\vect{h})]\bigr)=\fD_G(\dd\vect{h})\ ,
 \ee
 see e.g.~\cite[formula (5.1.51)]{Due19}. By using Field Independence of the time-ordered product (which is a further 
 renormalization condition), that is, 
 \be
 T_G\bigl(\exp_\ox(iF_G)\ox(\eps F_G)(\vect{h})\bigr)=-i(\eps S_G(F_G))(\vect{h})\ ,
 \ee
 the identity \eqref{eq:FE-2} is equivalent to
 \be\label{eq:SD}
 T_G\bigl(\exp_\ox(i F_G)\ox \fD_G(\dd\vect{h})\bigr)=S_G(F_G)\star\fD_G(\dd\vect{h})+i(\eps S_G(F_G))(\vect{h}),
 \ee
 which is the Schwinger-Dyson equation as given in \cite[formula (A.2)]{BF19} written for the Dirac field.

\paragraph{Equivalence of the relations Dynamics and Field equation.} 
To derive the relation Dynamics from the field equation  \eqref{eq:FE-2}, note the relations
\be\label{eq:aux1}
\frac{d}{d\la}\,F_G^{\la\vect{h}}=(\eps F_G^{\la\vect{h}})(\vect{h}),\quad
\frac{d}{d\la}\,\dl _{\la\vect{h}} L_G=\fD_G(\dd\vect{h})+2\la\langle\vect{h},\dd\vect{h}\rangle_G,
\ee
and
\be\label{eq:aux2}
(\eps \dl _{\la\vect{h}} L_G)(\vect{h})=\la\frac{d}{du}\Big\vert_{u=0}\fD_G(\dd\vect{h})^{u\vect{h}}
=2\la\langle\vect{h},\dd\vect{h}\rangle_G,
\ee
which follow from \eqref{eq:epsF}, \eqref{eq:dL0} and \eqref{eq:Diracfield}. Hence, setting
\be
K_G(\la)\doteq F_G^{\la\vect{h}}+\dl _{\la\vect{h}} L_G\in (G\ox\Floc)_\mathrm{even}
\ee
we obtain
\be\label{eq:aux3}
\frac{d}{d\la}\,K_G(\la)=\bigl(\eps K_G(\la)\bigr)(\vect{h})+\fD_G(\dd\vect{h}).
\ee
In addition we introduce
\be\label{eq:H}
U_G(\la)\doteq S_G(F_G)^{\star -1}\star S_G\bigl(K_G(\la)\bigr).
\ee
To obtain a simpler formula for $U_G(\la)$, we compute $\frac{d}{id\la}U_G(\la)$ by using \eqref{eq:aux3}:
\begin{equation}
\begin{split}
\frac{d}{id\la}U_G(\la) & = S_G(F_G)^{\star -1}\star
T_G\Bigl(\exp_{\ox}(iK_G(\la))\ox [\bigl(\eps K_G(\la)\bigr)(\vect{h})+\fD_G(\dd\vect{h})]\Bigr)\\
&= U_G(\la)\star R_G\Bigl(\exp_\ox(iK_G(\la)), [\bigl(\eps K_G(\la)\bigr)(\vect{h})+\fD_G(\dd\vect{h})]\Bigr),\label{eq:dH(la)}
\end{split}
\end{equation}
after insertion of the identity $S_G\bigl(K_G(\la)\bigr)\star S_G\bigl(K_G(\la)\bigr)^{\star -1}=1$ in the middle of the first line.
Now we insert the field equation \eqref{eq:FE-2} for the interaction $K_G(\la)$ and, in a second step, we take into account the relation \eqref{eq:star-K}:
\be
\frac{d}{id\la}U_G(\la)=U_G(\la)\star \fD_G(\dd\vect{h})=U_G(\la)\wedge \fD_G(\dd\vect{h}).
\ee
Since $U_G(0)=1$, we conclude that
\be\label{eq:H(la)}
U_G(\la)=\exp_\wedge(i\la\fD_G(\dd\vect{h}))=S_G(\dl_{\la\vect{h}} L_G),
\ee
where \eqref{eq:S(dL0)}  is inserted in the second equality. This identity can equivalently be written as
\be
S_G(F_G^{\vect{h}}+\dl_{\vect{h}} L_G)= S_G(F_G)\star S_G(\dl_{\vect{h}} L_G)=S_G(\dl_{\vect{h}} L_G)\star S_G(F_G)\ ,\label{eq:SD2}
\ee
the second equality follows from \eqref{eq:star-K}.
This is the ``off-shell'' version of the relation Dynamics \eqref{eq:dynamics} in terms of the perturbative $S$-matrix \eqref{eq:S-matrix}. 
More precisely, reducing the space of  field configurations to the solutions of the Dirac equation,
\be
  V_0\doteq \{h\in  V\,\vert\,\dd h=0\},
\ee
we have $\fD_G(\dd \vect{h})\vert_{V_0}=0$ and, hence, $ S_G(\dl_{\vect{h}} L_G)\vert_{V_0}=1$; that is, restricting the funtionals 
in this way, the relation \eqref{eq:SD} takes the  on-shell form of the relation Dynamics \eqref{eq:dynamics}.
 
That the field equation \eqref{eq:FE-2} follows from the relation Dynamics can easily be seen: applying 
$\frac{d}{id\la}\big\vert_{\la=0}$ to the relation dynamics in the form \eqref{eq:H(la)} and taking into account the formula \eqref{eq:dH(la)}, we obtain the field equation.
 
 \paragraph{Validity of the  further defining relations for the algebra $\mathfrak{A}_G$.} 
 The axiom Causality for the time-ordered product $T_{n,G}$
implies that $S_G(F_G)=T_G(\exp_\ox(iF_G))$ satisfies the Causal factorization \eqref{eq:caus-fact}. 
The validity of the further defining
relations for the algebra $\mathfrak{A}_G$ is obvious, in particular $S_G(F_G)^*=S_G\bigl((F_G)^*\bigr)^{\star -1}$ is a further 
renormalization condition for $T_{n,G}$, which can easily be satisfied. Summing up, the algebra
\be\label{eq:A-S}
\Ac\doteq  {\bigvee}_\star  \{S_G(F_G)\,\big\vert\, F_G\in (G\ox\Floc)_\mathrm{even}\}
\ee
(where $\bigvee_\star$ means the algebra, under the product $\star$, 
generated by members of the indicated set),  fulfills all defining relations for $\mathfrak{A}_G$. 

This can also be shown for the algebra obtained by the algebraic adiabatic limit \cite{BF0} of the relative $S$-matrices
\be
(S_G)_{F_G}(H_G)\doteq S_G(F_G)^{\star -1}\star S_G(F_G+H_G)\ ,\quad F_G,H_G\in (G\ox\Floc)_\mathrm{even}.
\ee
Again, the only non-trivial step 
is the verification of the relation Dynamics -- this can be done in precisely the same way as in \cite[Appendix]{BF19}.

\section{Conclusions and Outlook}
In this paper we have proposed a new description of theories with fermionic degrees of freedom, which is compatible with the $C^*$-algebraic framework introduced by \cite{BF19}. A key feature is the fact that only finite dimensional Grassmann algebras are needed in our construction, but the dependence on Grassmann parameters has to be functorial. This is very much in line with the language of locally covariant quantum field theory \cite{BFV} and shows the power of this, slightly more abstract, category theory viewpoint. The importance of the functorial formulation is also 
emphasized by \cite{Ll,HHS} in the treatment of supersymmetric theories. 
A potential future direction of research would be to apply our framework to 
some finite supersymmetric models, e.g. $N=4$ SYM.

In our future investigations, we plan to apply this framework to study gauge fields coupled to fermions, with the hope that we would be able to describe the chiral anomaly in the framework of \cite{BF19}. We addressed the issue of anomalies, at present only for scalar fields, in our paper \cite{BDFR21}.  Other possible applications include treatment of known exactly-solvable models including fermions, notably the Thirring model. In particular, we hope to be able to use the framework established in this work, together with the results of \cite{BFR17} to put the known duality between the sine Gordon model and the Thirring model into the $C^*$-algebraic framework of AQFT. 
\section*{Data availability statement}
Data sharing is not applicable to this article as no new data were created or analyzed in this study.
%%%%%%%%%%%%%%%%%%%%%%%%%%%%%%%%%%%%%%%%%%%%%%

\appendix

\section{Graded functionals}
For completeness, we include here the result on characterization of local functionals that depend on both fermionic and bosonic variables. Consider vector bundles $E_0\rightarrow M$ and $E_1\rightarrow M$, with their spaces of smooth sections $\Ecal_0\doteq  \Gamma(M, E_0)$ and $\Ecal_1\doteq  \Gamma(M, E_1)$.

Let $\Gamma'_{p|q}(M^{p+q}, E_0^{\boxtimes p}\boxtimes {E_1}^{\boxtimes q})$ denote the appropriate completion of the space ${\Gamma'(E_0)}^{\otimes_s p}\otimes \Gamma'(E_1)^{\wedge q}$,
understood as the space of distributional sections symmetric in the first $p$ and antisymmetric in the last $q$ arguments.

\begin{definition}\label{gradedFunctions}
	Define $\Ocal^k(\Ecal_0\oplus\Ecal_1[1])$ as the subspace of
	$C^\infty (\Ecal_0\times \wedge^k \Ecal_1,\CC)$ consisting of functionals that are  
	totally antisymmetric and $k$-linear in	the last $k$ arguments. Let $\Ocal(\Ecal_0\oplus\Ecal_1[1])\doteq  \prod_{k=0}^\infty \Ocal^k(\Ecal_0\oplus\Ecal_1[1])$.
\end{definition}
Derivatives with respect to the bosonic variable $\phi_0$ are defined in the usual way and derivatives with respect to the fermionc variable $\phi_1$ are given by Definition~\ref{GradedDerivaitve}, with $\phi_0$ fixed. In particular, for $F\in \Ocal^k(\Ecal_0\oplus\Ecal_1[1])$
\be
\frac{\delta^n F}{\delta \phi_0^n}(\phi_0)\in \Gamma_{n|k}'(M^{n+k},E_0^{\boxtimes n}\boxtimes E_1^{\boxtimes k})\cong \Gamma_{n|0}'(M^{n}, E_0^{\boxtimes n})\hat{\otimes}\Gamma_{0|k}'(M^{k}, E_1^{\boxtimes k})\,
\ee
so can be seen as a distribution with values in $\Ocal^k(\Ecal_1[1])$.
For proof see Theorem~III.10 of \cite{BDGR} and Proposition~3.4 of \cite{Book}. Similarly, for $n<k$, 
\be
\frac{\delta^n F}{\delta \phi_1^n}(\phi_0)\in \Gamma_{0|n}'(M^n, E_1^{\boxtimes n})\hat{\otimes} \Gamma_{0|k-n}'(M^{k-n}, E_1^{\boxtimes k-n})\,, 
\ee
so it is identified with a distribution with values in $\Ocal^{k-n}(\Ecal_1[1])$. Hence, in general, $\frac{\delta^n F}{\delta \phi_i^n}(\phi_0)$, $i=0,1$ is a distributional section on $M^n$ with values in $\Ocal(\Ecal_1[1])$. The usual rules for multiplication of distributions with given wave front sets apply in this case as well. More details can be found in \cite{Rej11a,Book}

\begin{theorem}
	Let $U$ be an open subset of $\Ec_0$ and 
	$F\in \Ocal^k(U\oplus\Ecal_1[1])$ be smooth in the sense of Bastiani.	Assume that
	\begin{enumerate}
		\item $F$ is additive.
		\item For every $\varphi\in U$, $h\in\bigoplus_{k\in\NN} \Ec_1^{\hat{\otimes} k-1}$,  the differentials
		$\frac{\delta F}{\delta \phi_0}(\varphi,h)$ and $\frac{\delta F}{\delta \phi_1}(\varphi,h)$ of $F$ have 
		empty  wave front sets and the maps 
		$(\varphi,h)\mapsto \frac{\delta F}{\delta \phi_0}(\varphi,h)$, $(\varphi,h)\mapsto \frac{\delta F}{\delta \phi_1}(\varphi,h)$ are Bastiani smooth from $U\times \bigoplus_{k\in\NN} \Ec_1^{\hat{\otimes} k}$ to $\Gamma_c(M,E^*_0)$ and  $\Gamma_c(M,E^*_1)$, respectively. Here $B_0^*$ and $B^*_1$ denote dual bundles.
	\end{enumerate}
	Then, for every  $\varphi\in U$,
	there is a neighborhood $V$ of the origin in $\Ec_0$, an integer $N$ and a smooth 
	$\CC$-valued function $\alpha$ on the $N$-jet bundle such that 
	\be
	F(\varphi+\psi;h_1\otimes\dots\otimes h_k)=\int_{M} \alpha(j^{i_0}_x(\psi),j^{i_1}_x(h_1),\dots,j^{i_k}_x(h_k))\,,
	\ee
	where $i_0,\dots,i_k<N$, for all $\psi\in V$ and $h\in\bigoplus_{k\in\NN} \Ec_1^{\hat{\otimes}k}$.
\end{theorem}
\begin{proof}
	Let $F\in \Ocal^k(\Ecal_0\oplus\Ecal_1[1])$, $k\neq 0$. The fundamental theorem of calculus implies that
	\be
	\begin{split}
	F(\varphi+\psi,h_1\otimes\dots\otimes h_k)
	 &=\int_0^1 dt  \int_M \frac{\delta F}{\delta \phi_0(x)}(\varphi+t\psi,h_1\otimes\dots\otimes h_k) \psi(x) dx\\
	&\,+\frac{1}{k}\sum_{i=1}^k(-1)^{k-1}\int_M   
	\frac{\delta F}{\delta \phi_1(x)}(\varphi,h_1\otimes\dots\widehat{h_i}\dots 
	\otimes h_k)(x) h_i(x) dx\\
	&=\int_0^1 dt  \int_M \frac{\delta F}{\delta \phi_0(x)}(\varphi+t\psi,h_1\otimes\dots\otimes h_k) \psi(x) dx\\
	&\qquad+ \int_M   \frac{\delta F}{\delta \phi_1(x)}(\varphi;h_2\otimes\dots \otimes h_k) h_1(x) dx\,,
	\end{split}
	\ee
	as $F(\varphi,0)=0$ and $F(\varphi,.)$ is totally antisymmetric. 
	Denote $h\doteq  h_1\otimes\dots\otimes h_k$.
	We apply lemma VI.13 of \cite{BDGR} to the first term and conclude that 
	for all $\varphi \in U$ and all 
	$\psi \in V $ such that the segment 
	$\varphi + t\psi \subset U$ for $0\le t\le 1$,
	\be
	F(\varphi+\psi,h)=\int_M c_{0,\psi,h}(x) dx+\int_M c_{1,\psi,h}(x) dx\,,
	\ee
	where
	\be
	c_{0,\psi,h}(x)= \int_0^1 \frac{\delta F}{\delta\varphi_0(x)}
	(\varphi+t\psi;h) \psi(x)dt
	\ee
	and
	\be
	c_{1,\psi,h}(x)= \frac{\delta F}{\delta\varphi_1(x)}
	(\varphi;h_2\otimes\dots \otimes h_k) h_1(x)\ .
	\ee
	Now, we apply proposition VI.14 of \cite{BDGR}
	and conclude that the functions  $c_{0,\psi,h}$ and  
	$c_{1,\psi,h}$ depend only on
	finite jets of $\psi$ and $h_1,\dots,h_k$. Finally, we use 
	Lemma~VI.15 to conclude that the resulting 
	function on the jet bundle is smooth. This concludes the proof. 
\end{proof}

%{\small 
\bibliographystyle{amsalpha}
%\bibliography{References}}

\end{document}